\documentclass[acmsmall]{acmart}

\usepackage{cleveref}
\newcommand{\hide}[1]{}
\usepackage{bussproofs}
\usepackage{datetime}
\usepackage{enumerate}
\usepackage{graphicx}
\usepackage{hyperref}
\usepackage{mathpartir}
\usepackage{stmaryrd}
\usepackage{tikz-cd}
\usepackage{xy}
\usepackage{comment}

\xyoption{all}

\theoremstyle{remark}
\newtheorem*{remark}{Remark}

\newcommand{\MCA}{\ensuremath{\mathcal A}}

\newcommand{\MCC}{\ensuremath{\mathcal C}}

\newcommand{\MCM}{\ensuremath{\mathcal M}}

\newcommand{\MCV}{\ensuremath{\mathcal V}}

\newcommand{\tj}[3]{#1 \vdash #2:#3}
\newcommand{\den}[1]{{\llbracket #1 \rrbracket}}
\newcommand{\sden}[1]{{[#1]}}

\newcommand{\op}{^{\mathrm{op}}}

\newcommand{\Set}{\mathbf{Set}}
\newcommand{\Cat}{\mathbf{Cat}}
\newcommand{\FinSet}{\mathbf{FinSet}}
\newcommand{\CCat}{\MCC}

\newcommand{\VCat}{\MCV}

\newcommand{\ACat}{\MCA}
\newcommand{\MCat}{\MCM}

\newcommand{\id}{\mathsf{id}}
\newcommand{\ob}[1]{|#1|}
\newcommand{\rev}{^\mathsf{rev}}
\newcommand{\Lan}{\mathsf{Lan}}
\newcommand{\res}{\mathsf{res}}
\newcommand{\sres}{\mathsf{sres}}
\newcommand{\yo}{\mathsf{y}}
\newcommand{\letin}[3]{\mathsf{let}\, #1\, \mathsf{be} \, #2\, \mathsf{in}\, #3}
\newcommand{\doin}[3]{\mathsf{do}\, #1 \leftarrow #2\, \mathsf{in}\, #3}

\newcommand{\ret}[1]{{\mathsf{return} \ #1}}

\newcommand{\monmult}{{\kappa}}
\newcommand{\monunit}{{\iota}}

\newcommand\sem[1]{\llbracket #1 \rrbracket}

\newcommand{\bind}{\mathrel{\scalebox{0.5}[1]{$>\!>=$}}}
\newcommand{\lunit}{(\,)}

\title[The Relative Monadic Metalanguage]
{The Relative Monadic Metalanguage}

\setcopyright{cc}
\setcctype{by}
\acmDOI{10.1145/3776702}
\acmYear{2026}
\acmJournal{PACMPL}
\acmVolume{10}
\acmNumber{POPL}
\acmArticle{60}
\acmMonth{1}
\received{2025-07-10}
\received[accepted]{2025-11-06}

\author{Jack Liell-Cock}
\orcid{0009-0005-7121-8095}
\email{jack.liell-cock@cs.ox.ac.uk}
\affiliation{%
  \institution{University of Oxford}
  \city{Oxford}
  \country{United Kingdom}
}

\author{Zev Shirazi}
\orcid{0009-0005-7578-2881}
\email{zev.shirazi@cs.ox.ac.uk}
\affiliation{%
  \institution{University of Oxford}
  \city{Oxford}
  \country{United Kingdom}
}

\author{Sam Staton}
\orcid{0000-0002-0141-8922}
\affiliation{%
  \institution{University of Oxford}
  \city{Oxford}
  \country{United Kingdom}
}

\ccsdesc[500]{Theory of computation~Categorical semantics}

\keywords{%
  relative monads,
  monad strength,
  finitary monads,
  graded monads,
  arrows.
}

\date{\today}
\begin{abstract}
    Relative monads provide a controlled view of computation. We generalise the monadic metalanguage to a relative setting and give a complete semantics with strong relative monads. Adopting this perspective, we generalise two existing program calculi from the literature. We provide a linear-non-linear language for graded monads, LNL-RMM, along with a semantic proof that it is a conservative extension of the graded monadic metalanguage. Additionally, we provide a complete semantics for the arrow calculus, showing it is a restricted relative monadic metalanguage. This motivates the introduction of ARMM, a computational lambda calculus-style language for arrows that conservatively extends the arrow calculus.
\end{abstract}
\begin{document}

\maketitle

\section{Introduction}
Over the past 15 years, \emph{relative monads} \cite{altenkirch_monads_2015} have emerged as a version of monads that provide a more controlled view of computational effects. They have applications in programming theory (e.g.~\cite{liell-cock_compositional_2025,katsumata_flexible_2022,DBLP:journals/pacmpl/JiangXN25}) as well as in more mathematical directions (such as~\cite{ arkor_formal_2024, arkor_nerve_2024,arkor_relative_2025,bourke_monads_2019,berger_monads_2012}).
Ordinary, non-relative monads have caused a dramatic transformation across programming languages over the past 35 years, both in theory (e.g.~\cite{katsumata_semantic_2005,pp-algop-geneff,pp-notions}) and in practice (e.g.~\cite{monadswork,scala}), as they support programming with effects in a principled, functional setting. A natural question is whether relative monads can cause a similarly profound transformation. This is the starting point for this paper.

One of the driving forces behind the adoption of ordinary monads in computer science is the monadic metalanguage, a mini programming language that precisely defines the relationships between pure functional programs and programs with effects. In this paper, we provide a \textbf{relative monadic metalanguage}, a version of the monadic metalanguage that is suited to relative monads. We show that this relative monadic metalanguage encapsulates previous calculi, namely the graded monadic metalanguage~(see \Cref{sec:gmm} and \cite{katsumata_parametric_2014}) and the arrow calculus~(see \Cref{sec:armm} and \cite{arrow-calc}). 

The key aspect of monads that supports their application in programming theory is the requirement of \emph{strength}. This is a critical component of the interpretation of the monadic metalanguage. Roughly, it allows monads and side-effecting computation to occur in context.
To enable a study of the relative monadic metalanguage, we also study the notion of strength in the relative case, relating notions from the literature~(\Cref{sec:comparing-srm}).

\hide{
Monadic programming is a powerful paradigm for modelling notions of computation in functional programming languages. Since the introduction of the monadic metalanguage~\cite{moggi_notions_1991}, monads have been used in a variety of practical languages (OCaml, Scala, Haskell) for effects such as I/O, statefulness and randomness. Programming languages that use monads to model effectful computation are amenable to structured analysis and formal reasoning~. Additionally semantic settings with monads~\cite{giry_categorical_1982,heunen_convenient_2017}
have influenced programming language design \cite{dash_affine_2023}. Monads have additionally been used to capture generalised notions of variable binding. They have been used in category theory to provide a framework for universal algebra, which has allowed several algebraic effects to be modelled using a monadic style of programming \cite{pp-algop-geneff,pp-notions,bourke_monads_2019}. }

\subsection{A relative monad primer}
Recall that monads consist of a type constructor $T$ together with methods for sequencing computations in $T$, typically a bind and return. In the ordinary monadic metalanguage,
if $A$ is a type, then $TA$ is a type of computation returning results in~$A$. But $TA$ is itself a type, just the same sort of thing as $A$, and so there is no problem with also writing $TTA$, the type of
computations that return results in the type of computations.

In many situations, this is useful, but in some situations, this higher-order type $TTA$ is hard to manipulate mathematically, and it is also conceptually difficult.
The idea of relative monads is to forbid $TTA$ by separating types into two different sorts. The sort of the type~$A$ (we call the sort $\ACat$) is different from the sort of the type $TA$ (we call the sort $\CCat$), so that $TTA$ is not allowed, because we cannot apply the monad type constructor to something of the monad type. 

For a concrete example, we consider the probability monad $D$ on the category of sets.
For a finite set~$A$, the set $DA$ is the set of probability distributions with $A$ outcomes,
\[\textstyle DA\quad=\quad\{p:A\to [0,1]~|~\sum_{a\in A}p(a)=1\}\text.\]
This set $DA$ is not finite, even though $A$ is finite. 
Although we \emph{can} extend $D$ to a constructor on all sets, the resulting set $DDA$ is not a very familiar concept from probability theory or statistics. Random measures such as the Dirichlet distribution \emph{are} common in statistics, but those typical examples are not contained in $DDA$.
Therefore, it is reasonable to want to focus on $DA$ where~$A$ is a finite set. (For more discussion of this example, see~\Cref{sec:finitary}.)

Generally, the setting for a relative monad is a
category $\ACat$ (e.g. finite sets), another category $\CCat$ (sets),
and a functor $J:\ACat \to \CCat$ (every finite set is, in particular, a set). 
Then a relative monad is also a functor $T:\ACat\to \CCat$, which enforces the discipline that one cannot apply~$T$ to itself. (See~\Cref{sec:monads-relative-monads} for the full definitions.)

\hide{Relative monads~\cite{altenkirch_monads_2015, arkor_formal_2024,arkor_nerve_2024,arkor_relative_2025} generalise the theory of monads by separating types into two sorts $\ACat$ types and $\CCat$ types, with a function, often called the \emph{root}, to move from the former to the latter. Computations can only be performed on $\ACat$ types, but the types of computations are $\CCat$ types. }

\subsection{Relative monadic metalanguages}
This paper revolves around the study of metalanguages for relative monads.  
We generalise ordinary monadic metalanguages to this relative setting and enforce this separation of types into two sorts in our new language. This leads to a bifurcation which generalises the monadic programming paradigm in two directions.
\begin{description}
    \item[Paradigm 1:] The first paradigm entails programs between $\CCat$ types. In this paradigm, a relative monad is a \textit{restriction} on where the computations may occur. In particular, a computation of a computation may not exist, so they may not be arbitrarily iterated. For example, $\ACat$ types may be finite types, while the effectful programs may produce infinite results of type $\CCat$. The finite distributions monad, which encodes probabilistic choice, is one such example. While the restriction to $\ACat$ may seem inhibitive, it greatly simplifies program reasoning and analysis.
    \item[Paradigm 2:] The second paradigm is concerned with programs between $\ACat$ types. For $\CCat$ types to have any utility, there must be a way to move from them back to the $\ACat$ universe. Formally, this is usually instantiated as an adjunction, for example, in Linear Logic~\cite{benton-lnl}, Call by Push Value~\cite{cbpv}, and the Enriched Effect Calculus~\cite{eec}. This paradigm permits computations that are beyond the expressiveness of the $\ACat$ types to be performed in $\CCat$, before being transported back to programs in $\ACat$. This allows a more general notion of computation since $\CCat$ types are not necessarily \textit{first-class}, i.e. the types of inputs and outputs to programs, but can model more general behaviour than those in $\ACat$.
\end{description}

In programming language theory, and in the monadic metalanguage, we are mainly concerned with \emph{strong monads}~\cite{kock_monads_1970, kock_strong_1972, dylan-tarmo-strong}. When sequencing effectful programs together, non-strong monads are not sufficient because they do not allow the threading of contexts into later programs. Strong monads provide a method for pairing a \emph{pure} context with a computational result to be passed into the next program.

To interpret a relative monadic metalanguage, it thus makes sense to define the notion of \emph{strong relative monads}~\cite{tarmo-grm,slattery-thesis}. However, there are subtleties to account for when making this generalisation. Many equivalent definitions of strong monads~\cite{dylan-tarmo-strong} are no longer equivalent when transported to the relative setting. Varying definitions of strong relative monads have been used to generalise and express a range of programming concepts~\cite{tarmo-grm, slattery-thesis, andrew_slattery_strong_nodate}. To get a holistic view of strength for relative monads, we collate the definitions and outline the relationships between them. We provide a hierarchy of subsumptions and specify the criteria for when the subsumptions flatten into equivalences. In particular, we re-establish the connection between strong monads and enriched monads in the relative setting (\Cref{thm:hat-strong-corresp}).
These results feed into our understanding of the relative monadic metalanguage and are crucial for providing semantic proofs about its applications.

\subsection{Applications of the Relative Monadic Metalanguage}

Using the theory we develop, we extend the relative monadic metalanguage to capture two concepts: graded monads and arrows. Both of these notions have been modelled by program calculi previously, but using the power of strong relative monads, we provide semantically backed program calculi for these concepts. Both lead to useful conservative extensions of the calculi in the previous literature.

\subsubsection{Graded monads}
\hide{Non-determinism is a powerful abstraction in programming, enabling concise expression of multiple possible outcomes. However, it can become computationally expensive when too many branches are explored. In practice, it is often useful to track an upper bound on the number of branches under consideration. For example, consider a self-replicating robot tasked with mapping a maze. At each step, the robot may choose to move left, duplicate and move both left and right, or replicate further and move in all four directions. However, if a robot attempts to move into a wall, it crashes and is lost. Since robots are costly, we are interested in maintaining an upper bound on the number of surviving robots throughout the exploration process. This scenario illustrates a form of graded non-determinism, where each computational step is annotated with a known upper bound on the number of possible outcomes. Graded monads provide a principled way to compose such computations, preserving this upper-bound information via type annotations throughout the computation.}

Graded monads~\cite{katsumata_parametric_2014,tate-graded,orchard-graded,fkm-graded,kura-graded,hicks-graded} are a generalisation of monads that give a semantic version of \emph{effect usage analysis}. They have been used to give semantic accounts of memory access analysis with the state monad, output analysis for the writer monad, and branch cuts in non-determinism (e.g.~\cite{katsumata_flexible_2022, katsumata_parametric_2014}). Moreover, they have been used in probabilistic programs to track non-deterministic choice~\cite{liell-cock_compositional_2025}, and to track execution traces to permit programmable inference~\cite{lew_trace_2020}. A big advantage of graded monads is that they can implement programs that are \textit{cost-sensitive}. A graded monadic metalanguage can be used for cost analysis of programs~\cite{cicek_cost-analysis_nodate} by keeping track of the resources in the grades. Graded monads have been established as an instance of enriched relative monads~\cite{mcdermott_flexibly_2022}, and in \Cref{thm:strong-graded-as-relative} we extend this connection in the strong case. The correspondence enables our development of a relative monadic style language for graded monads (which we call LNL-RMM). This language follows paradigm 2, where it is built on top of the linear-non-linear calculus~\cite{benton-lnl}. This calculus has two sorts of types, linear and non-linear, and constructs that allow transitions between them. Linear type theory is a \textit{resource-sensitive} type theory that doesn't allow discarding or duplication of variables in contexts. It is, therefore, suitable for reasoning directly with grades, since we often interpret grades as tracking of effect usage. In particular, for cost analysis of programs~\cite{cicek_cost-analysis_nodate}, the linearity here is used to ensure we cannot do any computation for free. We therefore explicitly have additional \textit{grade types} in the linear sort that allow us to directly program with grades. By encoding a strong relative monad over the linear to non-linear transition, we obtain a language for graded monads. In previous languages, gradings are explicit parameters, whereas ours are types within the language, and the graded monad is framed as a function from grades to a relative monad~\eqref{eqn:GMasRM}. These first-class grades enable more flexible programming and code reuse, as described in Section~\ref{sec:gmm}. 

\subsubsection{Arrows}
Arrows are a programming language paradigm that generalise monads~\cite{hughes-arrows}. While monads are a tool for sequencing effectful computations, arrows model computation that is not strictly sequential. They have a broad range of applications in functional programs, such as for parsers and printers~\cite{printing-parsing-1999}, web interaction~\cite{hughes-arrows}, circuits~\cite{Paterson2001}, graphical user interfaces~\cite{courtney2001}, and robotics~\cite{Hudak2003}. The established connection between arrows and strong relative monads~\cite{altenkirch_monads_2015, tarmo-grm} motivates our development of a relative monadic language for arrows (which we call ARMM). This language also follows paradigm 2, where the standard programs are between $\ACat$ types, while $\CCat$ types model parametric computations on inputs of type $\ACat$. The language we develop is a conservative extension of the arrow calculus ~\cite{arrow-calc} (\Cref{thm:conservative-armm}). It has all the expressiveness of arrows but additionally allows for variables of arrow computations, which promotes the reusability and flexibility of the programs. Moreover, it has more principled semantics, making it amenable to generalisations with additional type formers, other variants such as a linear language for arrows, or semantic proofs of language properties.

\subsection{Contributions}
Our main contributions can be summarised as follows:
\begin{itemize}
    \item We introduce the relative monadic metalanguage (which we call RMM) in \Cref{sec:strong-rmm} as a framework for languages in the paper and provide a semantics via strong relative monads.
    \item We introduce the language LNL-RMM as a more expressive language for reasoning with graded monads and provide a semantic proof in \Cref{thm:conservative-gmm} that it is a conservative extension of the graded monadic metalanguage. 
    \item We introduce the language ARMM as a more expressive language for programming with arrows and provide a semantic proof in \Cref{thm:conservative-armm} that it is a conservative extension of the arrow calculus.
\end{itemize}
Our framework calculus RMM and strong relative monads are introduced in \Cref{sec:strong-strong}, including an example calculus for reasoning with finitary monads. \Cref{sec:comparing-srm} proves a correspondence between strong relative monads and enriched relative monads and relates other notions of strong relative monads in the literature. \Cref{sec:gmm} introduces the LNL-RMM and proves conservativity by showing a correspondence between strong graded monads and strong relative monads in \Cref{thm:strong-graded-as-relative}. \Cref{sec:armm} introduces ARMM and proves conservativity by providing a complete semantics for the arrow calculus in \Cref{thm:completeness-arrow}.
\section{Monads, relative monads and a simple calculus} \label{sec:monads-relative-monads}
In this section, we review the definition of monads and relative monads and introduce a simple calculus (with unary contexts) for relative monads~\cite{moggi_notions_1991}. We present the definition of a monad in its \textit{Kleisli triple form} which aligns with their implementation in programming languages. 
\begin{definition} \label{def:monad}
    A \emph{monad} $(T, \eta, (-)^*)$ on a category $\CCat$ consists of:
    \begin{itemize}
        \item An object $TX\in \CCat$ for each $X\in \ob\CCat$,
        \item A morphism $\eta_X: X\to TX$ for each $X\in \ob\CCat$ called the \textit{unit}
        \item A collection of maps $(-)^*: \CCat(X,TY)\to \CCat(TX,TY)$ called the \textit{Kleisli extension} or \textit{bind}
    \end{itemize}
    such that for $f:X\to TY$ and $g:Y \to TZ$, the following equations hold:
    \begin{mathpar} \label{eqn:monad-equations}
        (\eta_X)^* = \id_{TX} \and
        f^* \circ \eta_X = f \and
        g^* \circ f^* = (g^* \circ f)^*
    \end{mathpar}
\end{definition}

From this data, we can define a unique functor structure on $T$ such that $(-)^*$ and $\eta$ are natural in $X, Y\in\ob{\CCat}$, given by $Tf = (\eta_Y \circ f)^*$ for $f:X\to Y$~\cite{marmolejo2010monads}. Monads are often presented in their \textit{monoid form} $(T, \eta, \mu)$~\cite{mac_lane_categories_1998}, where here $\mu_X=(\id_{TX})^*$. From the monoid form, we can recover Definition \ref{def:monad} by defining $f^*=\mu_Y \circ Tf$ for $f : X \to TY$~\cite{hoehnke_manes_1978}. 

\begin{definition} \label{def:relative-monad}
Let $\ob\ACat$ be a collection of objects, $\CCat$ a category, and $J: \ob{\MCA} \to \ob{\CCat}$ a mapping of objects, called the \emph{root}. A \textit{$J$-relative monad} $(T, \eta, (-)^*)$ consists of
\begin{itemize}
    \item An object $TA \in \ob{\CCat}$ for each $A\in\ob{\ACat}$
    \item A morphism $\eta_A : JA \to TA$ in $\CCat$ for each $A\in\ob{\ACat}$
    \item  A collection of maps $(-)^*: \CCat(JA,TB)\to \CCat(TA,TB)$
\end{itemize}
all such that for $f: JA\to TB$ and $g: JB \to TC$, the following equations hold:
\begin{mathpar} \label{eqn:relative-monad-equations}
    (\eta_A)^* = \id_{TA} \and
    f^*\circ \eta_A = f \and
    g^* \circ f^* = (g^* \circ f)^*
\end{mathpar}
A \textit{morphism of $J$-relative monads} $\gamma : (T,\eta,(-)^*) \to (S,\nu, (-)^\dagger)$ consists of maps $\gamma_A : TA \to SA$ satisfying, for any $f : JA \to TB$,
\begin{mathpar}
    \nu_A = \gamma_A \circ \eta_A, \and
    \gamma_B \circ f^* = (\gamma_B \circ f)^\dagger \circ \gamma_A.
\end{mathpar}
We denote the category of $J$-relative monads and morphisms as $\mathsf{Mon}(J)$.
\end{definition}
When $\ACat$ is a category with objects $\ob\ACat$ and $J: \ob\ACat \to \ob\CCat$ extends to a functor $J: \ACat \to \CCat$, the mapping on objects $T: \ob\ACat \to \ob\CCat$ has a unique functor structure such that $(-)^*$ and $\eta$ are natural in $A,B \in \ob{\ACat}$, defined by $Tf = (\eta_A \circ Jf)^*$. This additionally makes the morphisms of $J$-relative monads, $\gamma_A: TA \to SA$, natural in $A$.

\begin{remark}
    Relative monads are often defined with respect to roots that are \emph{functors} ~\cite{altenkirch_monads_2015,arkor_formal_2024}, rather than mappings on objects. Since any mapping on objects $J: \ob\ACat \to \ob\CCat$ can be seen as a functor from the discrete category on $\ob{\ACat}$, these definitions are equivalent. In \Cref{sec:enriched-relative}, we will consider the enriched setting, where discrete categories may not exist. Hence, this equivalence breaks, and there are benefits to considering mappings on objects. 
\end{remark}
\newcommand{\Mon}{\mathbf{Mon}}

\begin{example}
$\mathsf{Mon}(\id_\CCat)$ is the category of monads on $\CCat$ and monad morphisms (see \cite{barr_toposes_1985}).
\end{example}
\begin{example} \label{ex:restriction}
    If $(T,\eta,(-)^*)$ is a monad on $\CCat$, then we can define a relative monad $(T J,\eta^J,(-)^{*_J})$ by restricting the monad to objects under the image of $J$, where for $f: JA \to TJB$;
    \begin{mathpar}
        \eta^J_A = \eta_{JA} : JA \to TJA
        \and
        f^{*_J} = f^* : TJA \to TJB
    \end{mathpar}
    This assignment on objects extends to a functor $\res_J : \Mon(\id_\CCat)\to \Mon(J)$. We say $T$ \emph{restricts} to a $J$-relative monad $\widetilde T$ if $\widetilde T\cong \res_J {T}$.
\end{example}

Relative monads are often studied for roots that are \emph{well-behaved}~\cite{altenkirch_monads_2015}. In these cases, the root satisfies additional conditions so that any relative monad is the restriction of an ordinary monad.
When programming with effects presented by algebraic operations (see~\cite{pp-notions,pp-algop-geneff}), we can take $\ACat$ to be a category of \textit{arities}, so extensions of a relative monad $T$ for a well-behaved $J:\ACat\to\CCat$ can be viewed as a monad with arities in $\ACat$~\cite{berger_monads_2012,bourke_monads_2019}.

\begin{example} \label{ex:finitary}
    The inclusion $J : \FinSet \to \Set$ is well-behaved, and every $J$-relative monad is the restriction of an ordinary monad on $\Set$. The functor $\res_J$ defines an equivalence between $J$-relative monads and \textit{finitary} monads on $\Set$. 
\end{example}

\begin{example} \label{ex:yoneda-relative}
    The Yoneda embedding $\yo : \CCat \to \widehat{C}$ (where $\widehat{C}=[\CCat\op,\Set]$) is also well-behaved, and $\yo$-relative monads correspond to \emph{promonads}~\cite{borceux1994handbook} or \emph{weak arrows}~\cite{altenkirch_monads_2015}. These have been used to provide semantics for computational effects that cannot be captured by monadic style programming \cite{hughes-arrows, Rom_n_2023}.
\end{example}

\subsection{A unary context metalanguage for relative monads (URMM)} \label{sec:simple-language}
This section introduces a limited calculus for relative monads that only permits unary contexts. While this limits the applicability of the calculus, it provides a template for the more complex versions introduced later. In particular, in \Cref{sec:strong-strong}, we extend this calculus to admit product types and multiple-arity contexts. Fixing a category $\ACat$, the language URMM($\ACat$) has the following types.
\begin{alignat*}{3}
   &\text{Types:} \quad &X,Y &::= ~ JA ~|~ TA  \tag{$A\in \ob\ACat$}
\end{alignat*}
A judgement is of the form $\tj {x:X} u Y$, and terms are built inductively.
\begin{mathpar}
    \inferrule{ }{\tj {x:X} x {X}}
    \and
    \inferrule{\tj {x:X} {u} {JA}}{\tj {x:X} {\mathbf{f} \, u} {JB}}(f \in \ACat(A,B))
    \\
    \inferrule{\tj {x:X} u {JA}}{\tj {x:X} {\ret u} {TA}}
    \and
    \inferrule{\tj {x:X} u {TA} \and \tj {y:JA} t {TB}}{\tj {x:X} {\doin{y}{u}{t}} {TB}}
\end{mathpar}

In addition to imposing the standard reflexivity, symmetry, transitivity, substitution, and congruence laws for equality of terms (see \Cref{app:eqns}), we have equations:
\begin{mathpar}
     \inferrule{\tj {x:X} {u} {JA}}{\tj {x:X} {\mathbf{id_A} \; u = u} {JA}}
    \and
    \inferrule{\tj {x:X} {u} {JB}}{\tj {x:X} {\mathbf{(g\circ f)} \; u = \mathbf{g} \; ( \mathbf{f} \; u)} {JC}}(f \in \ACat(A,B), g\in \ACat(B,C))
    \and
    \inferrule{\tj {x:X} u {JA}
\and \tj {y:JA} t {TB}}{\tj {x:X} {\doin{y}{\ret{u}}{t} = t[u/y]} {TB}}
    \and
    \inferrule{\tj {x:X} u {TA}
}{\tj {x:X} {\doin{y}{u}{\ret y} = u} {TA}}
    \and
    \inferrule{\tj {x:X} u {TA} \and \tj {y:JA} t {TB} \and \tj {z:JB} v {TC}}{\tj {x:X} {\doin{z}{(\doin{y}{u}{t})}{v}= \doin{y}{u}{(\doin{z}{t}{v})}} {TC}}
\end{mathpar}
In categorical semantics, it is usual to interpret types as objects in a category, and terms as morphisms. The additional structure we impose here is interpreted as a relative monad.
\begin{definition}
    A \textit{model of URMM(\ACat)} consists of:
    \begin{enumerate}
        \item A category $\CCat$
        \item A functor $J : \ACat\to \CCat$
        \item A $J$-relative monad $(T,\eta,(-)^*)$
    \end{enumerate}
\end{definition}
Given a model of URMM(\ACat), $(\CCat,J,T)$, we define the interpretation of a judgement $\den {\tj {x:X} u Y}$ to be a morphism $\den X\to \den Y$ in $\CCat$ (where $\den X$ is defined in the evident way). We say an equation $\tj {x:X} {u=v} {Y}$ is valid in $(\CCat,J,T)$ if $\den{\tj {x:X} u {Y}}=\den{\tj {x:X} v {Y}}$. We inductively define the interpretation of terms as follows:
\begin{align*}
    \den{\tj {x:X} x X} &= \id_{\den{X}} \\
    \den{\tj {x:X} {\mathbf{f}\; u} {JC}} &= Jf \circ \den{\tj {x:X} u {JB}} \\
    \den{\tj {x:X} {\ret{u}} {TB}} &= \eta_B \circ \den{\tj {x:X} u {JB}} \\
    \den{\tj {x:X} {\doin{y}{u}{t}} {TC}} &= \den{\tj {y:JB} t {TC}}^* \circ \den{\tj {x:X} u {JB}}
\end{align*}
\begin{proposition} \label{prop:URMM-completeness}
    The models of URMM(\ACat) are a sound and complete semantics.
\end{proposition}
\begin{proof}
    The proof of soundness is done by structural induction and is standard. Completeness follows from the construction of a term model $(\mathcal{F},J,T)$. $\mathcal{F}$ has objects $JA$ and $TA$ for $A\in \ob\ACat$ and morphism $X\to Y$ are equivalence classes of terms $[\tj {x:X} u Y]$ under the equational theory. The functorial action of $J$ is defined by $[\tj {x:JA} {\mathbf{f} \; x} {JB}]$ while the $\mathsf{return}$ and $\mathsf{do}$ terms provide the data for the relative monad $T$.
\end{proof}

\section{Strong monads, strong relative monads and the relative monadic metalanguage} \label{sec:strong-strong}

Program contexts are often larger than a single variable. If we want to sum two numbers together, then we already require more than one variable. However, monads do not have sufficient structure for interpreting the sequencing of computations with multiple free variables. For example, say we intend to sequence the following two effectful programs.
\begin{mathpar}
    \tj {\Gamma} u {TA} \and \tj {\Gamma, x:A} t {TB}
\end{mathpar}
Interpreted as morphisms, these become
\begin{mathpar}
    \den{u}:\den{\Gamma} \to T\den{A},
    \and
    \den{t}:\den{\Gamma}\times \den{A} \to T\den{B}.
\end{mathpar}
To sequence these, we need to produce from $\den{t}$ a map $\den{\Gamma}\times T\den{A} \to T\den{B}$, i.e. apply Kleisli extension \emph{under the context $\Gamma$}. Strong monads provide the required semantics for such an operation. In the general setting, strong monads are defined with respect to monoidal products, of which the Cartesian product is an instance. We employ this setting for a linear type system (see~\cite{girard-ll,lnl-term-calculus,benton-lnl}) in \Cref{sec:gmm}, so we introduce the concepts at this generality.

\begin{definition}[\cite{mac_lane_categories_1998}]
    A \emph{monoidal category} $(\CCat,\otimes, I,\lambda,\rho,\alpha)$ consists of a category $\CCat$, an object $I \in \ob\CCat$, and a functor $\otimes: \CCat \times \CCat \to \CCat$, equipped with natural isomorphisms,
    \[
        \lambda_X : I\otimes X \to X, \qquad
        \rho_X : X\otimes I \to X, \qquad
        \alpha_{X,Y,Z} : (X\otimes Y) \otimes Z \to X\otimes (Y\otimes Z),
    \]
    called the \emph{left unitor}, \emph{right unitor} and \emph{asscociator}, satisfying coherence equations (see~\cite{mac_lane_categories_1998}).
    A monoidal category is \emph{symmetric} when there is an additional natural isomorphism $\sigma_{X,Y}: X\otimes Y \to Y\otimes X$, called the \emph{symmetry}, satisfying $\sigma_{Y,X}\circ \sigma_{X,Y} = \id_{X\otimes Y}$ and additional coherences (see~\cite{mac_lane_categories_1998}).
\end{definition}

A strong monad has the same data as a monad, but it can extend variables under the tensor.

\begin{definition}[\cite{moggi_notions_1991,dylan-tarmo-strong}] \label{def:str-monad}
    A \emph{strong monad} $(T, \eta, (-)^*)$ on a symmetric monoidal category $(\CCat, \otimes, I)$ consists of
    \begin{itemize}
        \item An object $TX\in \CCat$ for each $X\in \ob\CCat$
        \item A morphism $\eta_X: X \to TX$ for each $X\in \ob\CCat$
        \item A collection of maps $(-)^*: \CCat(\Gamma \otimes X,TY)\to \CCat(\Gamma\otimes TX,TY)$ natural in $\Gamma\in \ob{\CCat}$
    \end{itemize}
    such that, for $f:\Gamma \otimes X \to TY$ and $g: 
    \Delta \otimes Y \to TZ$, the following three equations hold:
    \begin{mathpar}
        (\eta_X \circ \lambda_{X})^* = \lambda_{TX} \and
        f^*\circ (\Gamma \otimes \eta_X) = f \and
        g^* \circ (\Delta \otimes f^*) \circ \alpha_{\Delta,\Gamma, TX} = (g^* \circ (\Delta\otimes f) \circ \alpha_{\Delta,\Gamma,X})^*
    \end{mathpar}
\end{definition}

Strong monads were originally presented using their \textit{monoid form}~\cite{kock_monads_1970,kock_strong_1972}. We discuss these equivalent definitions further in \Cref{sec:comparing-srm}. As in \Cref{def:relative-monad}, the \textit{Kleisli triple form} of a strong monad is more amenable to generalisation to the relative setting.

\begin{definition} \label{def:str-rel-monad}
    Let $\ob\ACat$ be a collection of objects, $(\CCat, \otimes, I)$ a symmetric monoidal category, and $J : \ob{\ACat} \to \ob{\CCat}$ a mapping on objects. A \emph{strong $J$-relative monad} $(T, \eta, (-)^*)$ on $\CCat$ consists of:
    \begin{itemize}
        \item An object $TA\in \ob{\CCat}$ for each $A \in \ob{\ACat}$,
        \item A morphism $\eta_A: JA \to TA$ for each $A \in \ob{\ACat}$,
        \item  A collection of maps $(-)^*: \CCat(\Gamma \otimes JA,TB)\to \CCat(\Gamma \otimes TA,TB)$ natural in $\Gamma \in \CCat$.
    \end{itemize}
    all such that for $f: \Gamma \otimes JA \to TB$ and $g: \Delta \otimes JB \to TC$, the following equations hold:
    \begin{mathpar}
        (\eta_A \circ \lambda_{JA})^* = \lambda_{TA} \and
        f^*\circ (\Gamma \otimes \eta_A) = f \and
        g^* \circ (\Delta \otimes f^*) \circ \alpha_{\Delta,\Gamma,TA} = (g^* \circ (\Delta \otimes f) \circ \alpha_{\Delta,\Gamma,JA})^*
    \end{mathpar}
    A \textit{morphism of strong $J$-relative monads}  $ \gamma : (T,\eta,(-)^*) \to (S,\nu, (-)^\dagger)$ consists of, for each $A \in \ob\ACat$, maps $\gamma_A: TA \to SA$ satisfying, for any $f : \Gamma \otimes JA \to TB$:
    \begin{mathpar}
        \nu_A = \gamma_A \circ \eta_A \and
        \gamma_B \circ f^* = (\gamma_B\circ f)^\dagger \circ (\Gamma \otimes \gamma_A)
    \end{mathpar}
    We write $\mathbf{SMon}(J)$ for the category of strong $J$-relative monads.
\end{definition}

\begin{remark}
    Other definitions of strong relative monad have been suggested in the literature~\cite{tarmo-grm, arkor_formal_2024}.
    We explore the connections between these definitions further in \Cref{sec:comparing-srm}. For non-symmetric
    monoidal categories, the number of inequivalent definitions fragments further. The most applicable notion for
    (non-commutative) linear type theories is that of \emph{bistrong} $J$-relative monads, which we describe in
    \Cref{app:bistrength}. This notion is closely related to \emph{strong multicategorical relative monads}~\cite{slattery-thesis}.
\end{remark}

Given a strong $J$-relative monad $(T, \eta, (-)^*)$, it has an \textit{underlying $J$-relative monad} $(T, \eta, (-)^{{*_U}})$, given by $f^{*_U} = (f \circ \lambda_{JA})^* \circ \lambda^{-1}_{TA}$ for $f \in \CCat(JA, TB)$. This extends $T$ to a functor when $J$ is a functor as in the non-strong case, and defines the object mapping for a forgetful functor $\mathbf{SMon}(J) \to \Mon(J)$. The restriction functor $\res_J: \Mon(\id_\CCat) \to \Mon(J)$ extends via this forgetful functor to a restriction $\sres_J: \mathbf{SMon}(\id_\CCat) \to \mathbf{SMon}(J)$.
 \[\begin{tikzcd}
	{\mathbf{SMon}(\mathsf{id}_\mathcal{C})} & {\mathbf{SMon}(J)} \\
	{\mathbf{Mon}(\mathsf{id}_\mathcal{C})} & {\mathbf{Mon}(J)}
	\arrow["{\mathsf{sres}_J}", from=1-1, to=1-2]
	\arrow[from=1-1, to=2-1]
	\arrow[from=1-2, to=2-2]
	\arrow["{\mathsf{res}_J}"', from=2-1, to=2-2]
\end{tikzcd}\]

\begin{example} \label{ex:set-rm-strength}
    When $\CCat=\Set$, every $J$-relative monad $(T,\eta,(-)^*)$ is uniquely strong, where the extension maps $(-)^\dagger : \Set(X\times JA,TB)\to \Set(X\times TA,TB) $ are given by $f^\dagger=\lambda(x,u).(f(x,-))^*(u)$.
\end{example}
\begin{example}
    $\mathbf{SMon}(\id_C)$ is the category of strong monads and strong monad morphisms~\cite{dylan-tarmo-strong}.
\end{example}

\subsection{The relative monadic metalanguage (RMM)} \label{sec:strong-rmm}

\newcommand{\rmm}{\lambda_{\mathsf{rmm}}}
Using strong relative monads, we now extend the unary context language URMM in \Cref{sec:simple-language} to a language RMM that admits product types and sequences computations under arbitrary contexts. It is analogous to the monadic metalanguage~\cite{moggi_notions_1991}, but restricts the type of sequencing allowed. It is also reminiscent of a fragment of Call-by-Push-Value~\cite{cbpv}, but the typical relative monads are not full Call-by-Push-Value models (see also~\cite{DBLP:journals/pacmpl/JiangXN25}).  In \Cref{sec:gmm} and \Cref{sec:armm}, we use this relative monadic metalanguage as a foundation for modelling generalisations of monadic effects.

Fixing a category $\ACat$, the language RMM($\ACat$) is defined as follows.
\begin{alignat*}{3}
   &\text{Types:} \quad &X,Y &::= 1 ~|~ X\times Y ~|~  JA ~|~ TA  \tag{$A\in \ob\ACat$} \\
   &\text{Contexts:} \quad &\Gamma &::= \emptyset ~|~ \Gamma,x:X 
\end{alignat*}
Terms are generated inductively. The full equational theory is given in \Cref{app:RMM-equations}.
\begin{mathpar}
    \inferrule{ }{\tj {\Gamma,x:X} x {X}}
    \and
    \inferrule{ }{\tj {\Gamma} \lunit {1}}
    \and
    \inferrule{\tj \Gamma u {X\times Y}}{\tj \Gamma {\pi_1 u} {X}}
    \and
    \inferrule{\tj \Gamma u {X\times Y}}{\tj \Gamma {\pi_2 u} {Y}}
    \and
    \inferrule{\tj {\Gamma} u {X} \and \tj {\Gamma} t {Y}}{\tj {\Gamma} {(u,t)} {X\times Y}}
    \and
    \inferrule{\tj \Gamma {u} {JA}}{\tj \Gamma {\mathbf{f} \, u} {JB}}(f \in \ACat(A,B))
    \and
    \inferrule{\tj \Gamma u {JA}}{\tj \Gamma {\ret{u}} {TA}}
    \and
    \inferrule{\tj {\Gamma} u {TA} \and \tj {\Gamma, x:JA} t {TB}}{\tj {\Gamma} {\doin{x}{u}{t}} {TB}}
\end{mathpar}
As in \Cref{sec:simple-language}, we interpret terms and types in a category, now with a \emph{strong} relative monad.
\begin{definition}
    A \textit{model of RMM(\ACat)} consists of:
    \begin{enumerate}
        \item A Cartesian category $\CCat$
        \item A functor $J : \ACat\to \CCat$
        \item A strong $J$-relative monad $(T,\eta,(-)^*)$
    \end{enumerate}
\end{definition}
If $\Gamma=a_1:A_1,..,a_n:A_n$, then its interpretation is given by the product $\prod_{i=1}^n \den{A_i}$ in~$\CCat$. A term $\tj \Gamma u X$ is interpreted as a morphism $\den \Gamma \to \den X$ in~$\CCat$. The only significant difference from the interpretation of terms in URMM(\ACat) is the sequencing terms.
\begin{displaymath}
    \den{\tj \Gamma {\doin{x}{u}{t}} {TB}}=(\den{\tj {\Gamma,x:JA} t {TB}})^*\circ (1_\Gamma,\den{\tj \Gamma u {TA}})
\end{displaymath}
\begin{theorem}
The models of RMM(\ACat) are a sound and complete semantics.
\end{theorem}
\begin{proof}
    As in \Cref{prop:URMM-completeness}, the proof of completeness is by constructing the term model (defined analogously to the term model in \Cref{prop:URMM-completeness}). The proof of soundness is also by structural induction, and by first verifying a substitution lemma $\den{\tj \Gamma {t[u/x]} Y}=\den{\tj {\Gamma,x:X} t Y}\circ (1_\den{\Gamma},\den{\tj \Gamma u X})$. Note that the equations in \Cref{def:str-rel-monad} for a strong relative monad reduce to the equations in \Cref{app:RMM-equations} on a Cartesian category.
\end{proof}

\subsection{Strongly finitary monads in the relative monadic metalanguage}\label{sec:finitary}

\newcommand{\coin}{\mathsf{coin}}
\newcommand{\Atoms}{\mathbb{A}}
\newcommand{\SVal}{\mathbb{V}}
\newcommand{\glkup}{{!}}
\newcommand{\gassgn}{\mathrel{:=}}
\newcommand{\gref}{\mathsf{ref}}

The relative monadic metalanguage can be thought of as a very simple
first-order programming language. The strength allows us to reason in context,
which is crucial for a programming language. The relative aspect
allows us to keep to a first order. Even if the ordinary monadic
metalanguage does not have explicit function types, it is
not really first-order, because it can talk about computations that return
computations ($TTA$). This type would be represented in Standard ML
with the type \texttt{unit -> (unit -> A)}, since functions of unit type are
thunks and \texttt{unit -> (-)} is a monad; notice that it is indeed higher
order there. 

This cut-down, first-order setting is ideal for equationally presenting
notions of computational effect, without worrying about aspects of higher-order
computation that are often orthogonal. The idea is that a 
computational effect can be presented by giving a first-order
programming language together with equations that it satisfies.
One can often add higher-order features separately, afterwards.

In this section, we demonstrate how this cut-down approach is already being taken in the literature, and can be
thought of as making essential use of an informal flavour of the relative monadic
metalanguage,
with examples from both probability and finitary $\Set$-monads (\Cref{sec:probfinitary}) and local store and more general finitary enriched monads (\Cref{sec:kp}). 

\subsubsection{Finitary monads, and examples from probabilistic programming}\label{sec:probfinitary}
Recall that a \emph{finitary monad} can be given by a relative monad
for the embedding $J:\FinSet\to \Set$, of finite sets into sets (\Cref{ex:finitary}).
These are necessarily strong (\Cref{ex:set-rm-strength}).
We can regard the objects of $\FinSet$
as being natural numbers.

For illustration, recall from the introduction the finitary monad of probability
distributions. The relative monad $D:\FinSet\to\Set$
has $Dn$ the probabilities on $n$ outcomes,
\[\textstyle Dn=\{(p_1\dots p_n)~|~p_i\geq 0\ \&\ \sum_{i=1}^n p_i=1\}\text.\]

Often, a finitary monad is presented by an algebraic theory with
operations and equations. In the case of $D$, this can be presented
by Stone's barycentric algebras~\cite{marshall_h_stone_postulates_1949}. 
This includes a binary operation $x\oplus y$, meaning `toss a fair coin and
choose $x$ or $y$', and axioms
\begin{equation}\label{eqn:stone-alg}
    x\oplus x=x,
    \qquad
    (u\oplus x)\oplus (y\oplus z)=
    (u\oplus y)\oplus (x\oplus z),
    \qquad
     x\oplus y=y\oplus x.
\end{equation}
If we restrict the $p_i$s to dyadic rationals, i.e. $p_i=\frac{a}{2^n}$, this is a full
presentation of $D$. This gives part of a correspondence between finitary monads and algebraic theories. 

On the other hand, the relative monadic metalanguage for $D$ supports
base types from $\ACat=\FinSet$, and a term
\[
  \tj {}\coin {T2}
\]
which is intuitively a coin toss operation, returning a boolean value
($2$), and is interpreted as the probability distribution
$(0.5,0.5)$.
This is a finite probabilistic programming language, similar to
DICE~\cite{dice}, which allows sequencing probabilistic operations and
calculating results.

In this setting, Stone's three equations~\eqref{eqn:stone-alg} become
\begin{equation} \begin{aligned}\label{eqn:stone-prog}
    \doin x \coin {\ret ()} &= \ret () \\
    \doin x \coin {\doin y \coin {\ret (x,y)}} &= \doin y \coin {\doin x \coin {\ret (x,y)}} \\
    \doin x \coin {\ret (\neg x)} &= \coin   
\end{aligned}
\end{equation}
where $\neg:2\to 2$ is the boolean negation. 

This program-based way of specifying the equations is complete in the
following sense. For any pair of programs $\tj {J\Gamma} {t,u} {Tn}$,
the equation $t=u$ is derivable from the laws of RMM(\ACat) if and only if it is true in the
model. 

This means that we can use program equations to directly specify a
finitary monad, instead of algebraic theories, as long as we restrict
to programs in the relative monadic metalanguage.   

As an aside, we note that the relative setting here is useful
specifically in probabilistic programming for at least three reasons.
\begin{enumerate}
\item Finite distributions extend to a monad $\widetilde D$ on
  $\Set$, the monad of finitely supported probability distributions.
  However, $\widetilde D\widetilde Dn$ does not contain the famous
  distributions-on-distributions from statistics, such as the
  Dirichlet distribution, because they are not finitely
  supported. There are other probability monads on other categories
  that do support the Dirichlet distribution~\cite{giry_categorical_1982}, but for this
  particular $D$ it is statistically strange to consider $\widetilde D\widetilde
  Dn$. 
\item The inference method in the DICE language works on fixed finite
  spaces. Thus, it is crucial to work with $Dn$ for some given finite
  $n$, rather than $DS$ for some potentially infinite set.
\item The synthetic approach to probability theory based on Markov
  categories can be formulated in terms of relative
  monads~\cite[\S 10.19]{fritz_synthetic_2020}. The restriction
  seems to be crucial for maintaining information flow axioms,
  which typically fail in models of a full monadic lambda calculus~\cite{dilations-fritz}. 
\end{enumerate}

\subsubsection{Kelly-Power style presentations, distributive Freyd categories, nominal sets, and local store}\label{sec:kp}
We have noted that the dyadic distributions monad could be presented by
three program equations~(\ref{eqn:stone-prog}). In general, if monads
are to be used for computational purposes, this is a helpful way to
specify them. This is a point made by \cite[p.~13]{pp-notions} and~\cite{staton-local-state}, where it is shown that, for instance, a local store monad (with dynamically allocated references) can be presented by program equations such as 
\begin{align}\label{eqn:localstateA}
x\gassgn y ; \glkup x\quad&=\quad
    x\gassgn y ; \ret y
    \\
    \doin x{\gref (u)}{x\gassgn v;\ret x}
    \quad&=\quad \doin x{\gref (v)}{\ret x}\label{eqn:localstateB}
\end{align}
We now explain how this kind of presentation essentially amounts to formal equational reasoning in the relative monadic metalanguage. A subtle point is that the dynamically allocated references do not belong to an ordinary finite set, and that leads us beyond finitary $\Set$-monads. We revisit this example more precisely at the end of this subsection. 

In fact, this general approach to presenting monads was essentially already taken by Kelly and Power~\cite[\S5]{kelly_adjunctions_1993} in their
formalization of algebraic presentations of strong monads, although
they did not write in terms of programs. We now recall some key aspects of that framework.

Kelly and Power begin (\S4) with a symmetric monoidal closed
  category $\CCat$ that is locally finitely presentable as a closed
  category. This has a full subcategory $\ACat=\CCat_{\mathrm{f}}$
  of finitely presentable objects.
  In general, a finitary enriched monad on $\CCat$ is, in our terminology,
  a strong monad relative to the embedding $J:\ACat\hookrightarrow
  \CCat$.
  A small worry is that $\ACat$ is a large category, and so they pick
  a set $N$ of representatives of each isomorphism class. These might
  be the basic types in our language.
  In the case of finitary monads on $\Set$, $\ACat=\FinSet$, and $N$
  can be the set of natural numbers.

  To present the theory, Kelly and Power begin with a
  signature, which is an assignment $B:N\to \CCat$. This is
  picking out for each type $A\in N$, a $\CCat$ object of basic
  computations of type $A$. This is not a functor.
  For probability, we would put $B2=\{\coin\}$, and $Bn=\emptyset$
  for $n\neq 2$. But in general, $N$ might not range over natural numbers.

  Kelly and Power then generate a relative monad $FB$ from $B$. 
  This is, in effect, terms built from an RMM calculus specialized to $\CCat$, 
  starting with basic operations $x:Jm \vdash t: FBn$ for each $n$.

  Kelly and Power then allow a system of equations, which are
  essentially 
  pairs from $FB$. They show that these generate an enriched finitary
  monad on $\CCat$, and that moreover, every enriched finitary monad arises in
  this way. 

  The approach in~\cite{kelly_adjunctions_1993} is focused on locally finitely
  presentable categories, but the procedure also goes through with
  other
  sound limit doctrines (e.g.~\cite{sam-freyd-lawvere}), giving complete approaches to presenting
  Freyd categories and distributive Freyd categories.
  These are theories of first-order programming
  languages with side effects, and
  RMM($\ACat$) for distributive Freyd categories is precisely the fragment of
  FGCBV~\cite{levy-environments} without function types. 

  For a precise example, we consider the presentation of a monad for local store
  in~\cite{staton-local-state}. There, the category $\CCat$ comprises Pitts' nominal
  sets~\cite{pitts-nom}, and $\ACat$ comprises the full subcategory of finitary strong nominal sets~\cite[Thm.~2.7]{pitts-nom}.
  We do not need the full definitions, but the important point is that
  $\CCat$ is like a category of sets
  but with a special object $\mathbb{A}$ of atoms, and $\ACat$ is a subcategory spanned by interpretations of first-order types. 
  Then~\cite{staton-local-state} presents local store as a strong monad relative to the
  inclusion $\ACat\to\CCat$. 
  These will be used
  as names of reference cells, and the equivariance of nominal sets maintains
  that they are abstract -- there is no pointer arithmetic.
  We also assume a strong nominal set~$\SVal$ of storable values. 

  The theory is generated by basic programs,
  \[
    x:J\Atoms,y:J\SVal\vdash x \gassgn y : T1
    \qquad
    x:J\Atoms\vdash \glkup x : T\Atoms
    \qquad
    y:J\SVal\vdash \gref( y) : T\Atoms
  \]
  for assign, lookup, and allocating a new reference cell, with equations such as \eqref{eqn:localstateA} and~\eqref{eqn:localstateB}.
  The paper discusses completeness of the system of 14 equations and shows that they present a monad,
  reasoning using a language that is essentially an instantiation of RMM. In equation~\eqref{eqn:localstateA}, the strength allows us to use the variable $x$ twice on the left-hand side, and $y$ twice on the right.

\section{Comparing notions of strong relative monads} \label{sec:comparing-srm}

Strong monads have several equivalent presentations~\cite{dylan-tarmo-strong}; however, these equivalences are not all maintained when they are transported to the relative setting. In this section, we recall notions of strong monads and describe their corresponding generalisations. We show when equivalent notions are preserved, and provide criteria for recovering them when they are not. Understanding these correspondences is crucial in the development of refined semantic type theories extending RMM.  In particular, our results about graded monads in \Cref{sec:gmm} rely on connections with \emph{enriched relative monads}~\cite{mcdermott_flexibly_2022, arkor_formal_2024}. Additionally, our semantics of the arrow calculus in \Cref{sec:armm} uses a weaker notion of strength.

We begin by recalling the categorical background for these correspondences.
\begin{definition}[\cite{kelly}]
   Let $(\VCat,\otimes,I)$ be a symmetric monoidal category. A $\VCat$-enriched category $\CCat$ consists of:
    \begin{itemize}
        \item A collection of objects $\ob\CCat$
        \item A $\VCat$-object of morphisms $\CCat[X,Y]$ for each $X,Y\in\ob\CCat$
        \item An identity morphism $\mathbf{id}_X: I \to \CCat[X,X]$ for each $X\in \ob\CCat$
        \item Composition morphisms $
        \bullet_{X,Y,Z}: \CCat[Y,Z] \otimes \CCat[X,Y] \to \CCat[X,Z]$ for $X,Y,Z\in\ob\CCat$
    \end{itemize}
    such that composition and identities satisfy associativity and unitality equations (see~\cite{kelly}).
    
    A \emph{$\VCat$-enriched functor} $J : \ACat \to \CCat$ between $\VCat$-enriched categories $\ACat$ and $\CCat$ consists of a mapping on objects $J : \ob{\ACat}\to \ob{\CCat}$ along with morphisms $J_{A,B} : \ACat[A,B]\to \CCat[JA,JB]$ in $\VCat$ preserving composition and identities.
\end{definition}
\newcommand{\daytimes}{\mathbin{\widehat\otimes}}

\begin{definition}[Day monoidal structure~\cite{day_closed_1970}] \label{def:day}
    Let $(\CCat, \otimes, I)$ be a small symmetric monoidal category. Then $(\widehat{\CCat}, \daytimes, \widehat I)$ is a symmetric monoidal closed category (i.e. $-\daytimes F \dashv F\multimap -$ for all $F\in\ob{\widehat{\CCat}}$), and in particular, is self-enriched. Here, $\widehat I = \yo I_{\CCat}$  and
    \[
    (F \daytimes G)(X) = \int^{Y,Z\in\CCat\op} \CCat(X, Y \otimes_\CCat Z) \times FY \times GZ.
    \]
\end{definition}
\begin{definition}[Fully faithful enrichment] \label{def:ffe}
If $J : \ob{\ACat} \to \ob{\CCat}$ is a mapping on objects and $\CCat$ is a $\VCat$-enriched category, the \emph{$J$-fully faithful} enrichment of $\ACat$ and $J$ is defined by $\ACat[A,B] = \CCat[JA,JB]$ and $J_{A,B} = \id_{\CCat[JA,JB]}$.
\end{definition}
These definitions allow us to present the following three notions, which give varying perspectives on strong monads. The first provides a natural semantics for monadic programming languages by allowing extension under any computational context. The second is a categorical viewpoint; it shows that a strong monad is analogous to the data for a monad in monoid form, but with additional coherent structure. The final perspective bridges the gap between these two; it says that strength is an \textit{internalisation} of the extension operation of a monad. This is necessary for programming languages, where the terms of the language need to \textit{internalise} the usage of bind to program with it.

\begin{proposition}[\cite{dylan-tarmo-strong}] \label{prop:equiv-strong-monad}
    Consider a small symmetric monoidal category $\CCat$. The following data are in bijective correspondence:
    \begin{enumerate}
        \item A strong monad on $\CCat$, as in Definition \ref{def:str-monad}.
        \item A monad on $\CCat$ with strength maps $\theta_{X,Y} : X\otimes TY \to T(X\otimes Y)$ satisfying compatibility conditions with the monoidal and monad structures (see ~\cite{kock_strong_1972}, also \Cref{def:strength-maps} with $J = \id_\CCat$).
        \item A $(\widehat{\CCat},\daytimes,\widehat{I})$-enriched monad (see~\cite{dylan-tarmo-strong}, also \Cref{def:enriched-relative} with $J = \id_\CCat$) on $\CCat$ with the $\yo$-fully faithful enrichment.
    \end{enumerate}
\end{proposition}
In \Cref{sec:enriched-relative} and \Cref{sec:j-strengths}, we provide analogues for these three perspectives for relative monads. We show that while the first and third data still coincide, the second is no longer generally equivalent. It is weaker and not sufficient to provide semantics for RMM.

\subsection{Connection with enriched relative monads} \label{sec:enriched-relative}

In higher-order languages, strong monads yield an operation
$(A\to TB)\to (TA \to TB)$
providing an \textit{internalisation} of the Kleisli extension into the language. Such internalisations are given semantics using \textit{enriched category theory}. In this section, we connect our notion of strong relative monads (\Cref{def:str-rel-monad}) to enriched relative monads.

\begin{definition}[\cite{arkor_formal_2024}] \label{def:enriched-relative}
    Let $(\VCat,\otimes, I)$ be a symmetric monoidal category, $\ob\ACat$ a collection of objects, $\CCat$ a $\VCat$-enriched category, and $J : \ob{\ACat} \to \ob{\CCat}$ a mapping of objects. A \emph{$\VCat$-enriched $J$-relative monad} $(T,\eta, (-)^*)$, consists of:
    \begin{itemize}
        \item An object $TA\in\ob\CCat$ for each $A\in\ob\ACat$,
        \item A morphism $\eta_A: I \to \CCat[JA, TA]$ in $\VCat$ for each $A\in\ob\ACat$,
        \item A family of morphisms $*: \CCat[JA, TB] \to \CCat[TA, TB]$ in $\VCat$
    \end{itemize}
    making the following diagrams commute.
    \begin{mathpar}
        \begin{tikzcd}
    	I & {\CCat[JA,TA]} \\
    	& {\CCat[TA, TA]}
    	\arrow["{\eta_A}", from=1-1, to=1-2]
    	\arrow["{\mathbf{id}_{TA}
        }"', from=1-1, to=2-2]
    	\arrow["{*}", from=1-2, to=2-2]
        \end{tikzcd}
        \and
        \begin{tikzcd}
    	{\CCat[JA,TB]\otimes I} & {\CCat[TA,TB]\otimes \CCat[JA,TA]} \\
    	& {\CCat[JA,TB]}
    	\arrow["{*\otimes \eta_A}", from=1-1, to=1-2]
    	\arrow["{\rho_{\CCat[JA,TB]}}"', from=1-1, to=2-2]
    	\arrow["\bullet_{JA,TA,TB}", from=1-2, to=2-2]
        \end{tikzcd}
        \and
        \begin{tikzcd}[column sep=large]
    	{\CCat[JB,TC]\otimes \CCat[JA,TB]} && {\CCat[TB,TC]\otimes \CCat[TA,TB]} \\
    	{\CCat[TB,TC]\otimes \CCat[JA,TB]} & {\CCat[JA,TC]} & {\CCat[TA,TC]}
    	\arrow["{*\otimes *}", from=1-1, to=1-3]
    	\arrow["{*\otimes \CCat[JA,TB]}"', from=1-1, to=2-1]
    	\arrow["\bullet_{TA,TB,TC}", from=1-3, to=2-3]
    	\arrow["\bullet_{JA,TB,TC}"', from=2-1, to=2-2]
    	\arrow["{*}"', from=2-2, to=2-3]
        \end{tikzcd}
    \end{mathpar}
    A \textit{morphism of $\VCat$-enriched $J$-relative monads} $\gamma : (T,\eta,(-)^*) \to (S,\nu, (-)^\dagger)$ consists of, for each $A \in \ob\ACat$, maps $\gamma_A: I \to \CCat[TA, SA]$ satisfying compatibility conditions~\cite{arkor_formal_2024}.
\end{definition}

As in the unenriched setting, when $\ACat$ is an enriched category with objects $\ob\ACat$ and $J: \ob\ACat \to \ob\CCat$ extends to an enriched functor $J: \ACat \to \CCat$, $T$ extends to an enriched functor $T: \ACat\to \CCat$, where $T_{A,B}= *\circ \bullet_{JA,JB,TB} \circ (\eta_A \otimes \CCat[JA,JB])\circ \lambda^{-1}_{\CCat[JA,JB]} \circ J_{A,B}$.

\begin{theorem} \label{thm:hat-strong-corresp}
    Let $\ob\ACat$ be a collection of objects, $\CCat$ a small symmetric monoidal category, and $J: \ob\ACat \to \ob\CCat$ a mapping on objects. Then there is an isomorphism of categories between:
    \begin{enumerate}
        \item The category of $(\widehat{\CCat},\daytimes,\widehat{I})$-enriched $J$-relative monads, via the $\yo$-fully faithfully enrichment of $\CCat$.
        \item The category of strong $J$-relative monads.
    \end{enumerate}
\end{theorem}

\begin{proof}
    By the Yoneda lemma, there is an isomorphism $\CCat[X,Y]\cong \CCat(- \otimes X, Y)$. The bijection on the units of the monad is thus given by:
    \[
        \widehat{\CCat}(\widehat{I}, \CCat[JA, TA]) \cong \widehat{\CCat}(\yo I_{\CCat}, \CCat(-\otimes JA, TA)) \cong \CCat(I_{\CCat}\otimes JA, TA) \cong \CCat(JA, TA)
    \]
    The bijection on the Kleisli extensions is given by:
    \[
        \widehat{\CCat}(\CCat(-\otimes JA, TB), \CCat(-\otimes TA, TB)) \cong \widehat{\CCat}(\CCat[JA,TB], \CCat[TA, TB])
    \]
    The coherence conditions lift through these isomorphisms accordingly, and morphisms additionally correspond.
\end{proof}

We recall that there is a forgetful functor $\VCat(I,-)_\# : \VCat\Cat \to \Cat$ defined by applying the functor $\VCat(I,-)$ to the $\VCat$-objects of morphisms~\cite{kelly}. In the case of $\widehat{\CCat}$-enrichment, $\widehat{\CCat}(\widehat{I},-)$ is evaluation on $I$. The mapping from (1) to (2) in Theorem \ref{thm:hat-strong-corresp} can be understood to be applying this functor to the data of a $\widehat{\CCat}$-enriched relative monad, as it sends $(T,\eta,(-)^*)\mapsto (T,\eta_I(1_I),(-)^*_I)$. In particular, if $\ACat$ is an $\widehat{\CCat}$-enriched category with objects $\ob{\ACat}$ and $\widehat{J}$ is an $\widehat{\CCat}$-enriched functor, this mapping sends an enriched $\widehat{J}$-relative monad to a strong $\widehat{\CCat}(\widehat{I},-)_{\#}\widehat{J}$-relative monad. Since $J$ and $\ACat$ are not uniquely enriched, the correspondence in \Cref{thm:hat-strong-corresp} does not extend to the data of the induced functor on $T$. However, if we fix a choice of enrichment, we recover the correspondence for the induced functors, partly extending Proposition \ref{prop:equiv-strong-monad}.

\begin{corollary} \label{cor:functorial-enriched-correspondence}
    Let $(\CCat, \otimes, I)$ be a small symmetric monoidal category, $\ACat$ a $\widehat{\CCat}$-enriched category, and $\widehat{J} : \ACat \to \CCat$ a $\widehat\CCat$-enriched functor. Then, the category of $\widehat\CCat$-enriched $\widehat{J}$-relative monads is isomorphic to the category of strong $\widehat{\CCat}(\widehat{I},-)_{\#}\widehat{J}$-relative monads.

    Moreover, this isomorphism preserves the functor structure induced by $\widehat{J}$ and $\widehat{\CCat}(\widehat{I},-)_{\#}\widehat{J}$ for the enriched and strong relative monads, respectively.
\end{corollary}

\begin{example}
    If $J$ is a fully faithful functor, then the $J$-fully faithful enrichment of $\ACat$ and $J$ gives a $\widehat{J}$ satisfying $J=\widehat{\CCat}(\widehat{I},-)_{\#}\widehat{J}$ and so \Cref{cor:functorial-enriched-correspondence} applies.
\end{example}

There is a left adjoint $F_{\widehat{\CCat}} \dashv \widehat{\CCat}(\widehat{I},-)_{\#}$, given on objects by sending a category $\ACat$ to a $\widehat{\CCat}$-enriched category with the same objects and with $(F_{\widehat{\CCat}}\ACat)[A,B] = \ACat(A,B)\times \CCat(-,I)$. Hence, functors $J : \ACat \to \CCat$ correspond to $\widehat{\CCat}$-enriched functors $\widehat{J} : F_{\widehat{\CCat}}\ACat \to \CCat$, and so we can obtain a similar result in the other direction.
 
\begin{corollary} \label{cor:functorial-enriched-correspondence-2}
    Let $\CCat$ be a small symmetric monoidal category, and let $J : \ACat \to \CCat$ be any functor and $\widehat{J} : F_{\widehat{\CCat}}\ACat \to \CCat$ be the transpose under the adjunction $F_{\widehat{\CCat}} \dashv \widehat{\CCat}(\widehat{I},-)_{\#}$. Then, the category of $\widehat\CCat$-enriched $\widehat{J}$-relative monads is isomorphic to the category of strong $J$-relative monads.

    Moreover, when the relative monads are given the functor structure induced by $\widehat{J}$ and $J$, respectively, this isomorphism coincides with the transpose under the adjunction $F_{\widehat{\CCat}} \dashv \widehat{\CCat}(\widehat{I},-)_{\#}$.
\end{corollary}

\begin{example}
  Since $J=\widehat{\CCat}(\widehat{I},-)_{\#}\widehat{J} \circ \iota_\ACat$, where $\iota$ is the unit of the adjunction $F_{\widehat{\CCat}} \dashv \widehat{\CCat}(\widehat{I},-)_{\#}$, we cannot directly obtain the functor structure of any strong $J$-relative monad by applying the forgetful functor $\widehat{\CCat}(\widehat{I},-)_{\#}$. However, if $I$ is initial or terminal, or more generally $\CCat(I, I)=1$, the unit of the adjunction $F_{\widehat{\CCat}} \dashv \widehat{\CCat}(I,-)_{\#}$ is an isomorphism. Then, the correspondence in \Cref{thm:hat-strong-corresp} extends to functors via \Cref{cor:functorial-enriched-correspondence-2}.
\end{example}

When a category has enough structure, rather than internalising the Kleisli extension in its category of presheaves, we can internalise the structure directly.

\begin{theorem} \label{thm:closed-enriched-strong-corresp}
    Let $\CCat$ be a symmetric monoidal closed category. Then $\CCat$ is self-enriched, and there is an isomorphism of categories between the category of $\CCat$-enriched $J$-relative monads and the category of strong $J$-relative monads.
\end{theorem}
\begin{proof}
    We can choose a suitable universe of sets so that $\CCat$ is small. The Yoneda embedding $\yo : \CCat \to \widehat\CCat$ is a fully faithful, strong closed symmetric monoidal functor. In particular, $\yo X \multimap \yo Y \cong \yo (X\multimap Y)$, and the data of a $\widehat{\CCat}$-enriched $J$-relative monad unravels via the Yoneda lemma to the data of a $\CCat$-enriched $J$-relative monad. The result then follows from Theorem~\ref{thm:hat-strong-corresp}.
\end{proof}
\subsection{Refined notions of strengths for relative monads} \label{sec:j-strengths}
The notion of strength introduced in \Cref{def:str-rel-monad} validates the following rule of RMM.
\newcounter{infrule}
\newcommand{\ruletag}[1]{%
  \refstepcounter{infrule}%
  \label{#1}%
  \text{(T-bind)}%
}
\begin{mathpar}
    \inferrule{\tj {\Gamma} u {TA} \and \tj {\Gamma, x:JA} t {TB}}{\tj {\Gamma} {\doin{x}{u}{t}} {TB}}
    \ruletag{rule:sequencing}
\end{mathpar}

This allows us to sequence computations in arbitrary contexts $\Gamma$. However, for some languages using the relative monadic metalanguage framework, it is sufficient to validate \hyperref[rule:sequencing]{T-bind} for restricted contexts. For example, the nominal language for local state discussed in \Cref{sec:finitary} only requires \hyperref[rule:sequencing]{T-bind} for contexts of shape $\Gamma = x_1:JA_1,\dots,x_n:JA_n$. Furthermore, in \Cref{sec:armm}, we show that the arrow calculus~\cite{arrow-calc} is also a relative monadic language that restricts the shape of the context for sequencing computations. This motivates the following definition, which validates \hyperref[rule:sequencing]{T-bind}  for contexts of shape $\Gamma=x_1: WM_1, \dots, x_n: WM_n$.

\begin{definition}
    Let $\ob{\ACat}$ be a collection of objects, $\CCat$ and $\MCat$ be symmetric monoidal categories, and $J : \ob{\ACat}\to \ob{\CCat}$ be a mapping of objects. If $(W,\kappa,\iota) : \MCat \to \CCat$ is a strong symmetric monoidal functor~\cite{mac_lane_categories_1998}, a \emph{$W$-strong $J$-relative monad} consists of:
 \begin{itemize}
        \item An object $TA\in \ob{\CCat}$ for each $A \in \ob{\ACat}$
        \item A morphism $\eta_A: JA \to TA$ for each $A \in \ob{\ACat}$
        \item  A collection of maps $(-)^*: \CCat(W\Gamma \otimes JA,TB)\to \CCat(W\Gamma \otimes TA,TB)$ natural in $\Gamma \in \ob\MCat$
    \end{itemize}
    all such that for $f: W\Gamma \otimes JA \to TB$ and $g: W\Delta \otimes JB \to TC$, the following equations hold:
    \begin{mathpar}
        (\eta_A \circ \lambda_{JA} \circ (\iota^{-1}\otimes JA))^* = \lambda_{TA} \circ(\iota^{-1}\otimes TA) \and
        f^*\circ (W\Gamma \otimes \eta_A) = f \and
        g^* \circ (W\Delta \otimes f^*) \circ \alpha_{W\Delta,W\Gamma,TA} \circ (\kappa^{-1}_{\Delta,\Gamma}\otimes TA) = (g^* \circ (W\Delta \otimes f) \circ \alpha_{W\Delta,W\Gamma,JA} \circ (\kappa^{-1}_{\Delta,\Gamma}\otimes JA))^*
    \end{mathpar}
    A \emph{morphism of $W$-strong $J$-relative monads}  $ \gamma : (T,\eta,(-)^*) \to (S,\nu, (-)^\dagger)$ consists of, for each $A \in \ob\ACat$, maps $\gamma_A: TA \to SA$ satisfying, for any $f : W\Gamma \otimes JA \to TB$:
    \begin{mathpar}
        \nu_A = \gamma_A \circ \eta_A \and
        \gamma_B \circ f^* = (\gamma_B\circ f)^\dagger \circ (W\Gamma \otimes \gamma_A)
    \end{mathpar}
    We write $W\mathbf{SMon}(J)$ for the category of $W$-strong $J$-relative monads.
\end{definition}

\begin{example}
    When $\MCat=\CCat$ and $W=\id_\CCat$, we recover the definition of strong relative monad in \Cref{def:str-rel-monad}.
\end{example}
As with \Cref{def:str-rel-monad}, a $W$-strong $J$-relative monad has an underlying $J$-relative monad $(T, \eta, (-)^{*_U})$, given by $f^{*_U} = (f \circ \lambda_{JA}\circ (\iota^{-1}\otimes JA))^* \circ  (\iota \otimes TA) \circ\lambda^{-1}_{TA}$ for $f \in \CCat(JA, TB)$, which defines the object mapping for a forgetful functor $W\mathbf{SMon}(J) \to \Mon(J)$. Furthermore, there is a forgetful functor $U_W: \mathbf{SMon}(J) \to W\mathbf{SMon}(J)$ that provides a factorisation for the forgetful functor $\mathbf{SMon}(J) \to \mathbf{Mon}(J)$. Under additional conditions, $U_W$ is an isomorphism and so $W$-strong $J$-relative monads coincide with strong $J$-relative monads.
\begin{proposition} \label{prop:J-strength-is-strong}
    Assume that $\CCat$ is a symmetric monoidal category such that $X\otimes -$ preserves colimits for all $X \in \ob{\CCat}$. If $W: \MCat \to \CCat$ is a dense strong symmetric monoidal functor, then $U_W : \mathbf{SMon}(J) \to W\mathbf{SMon}(J)$ is an isomorphism of categories.
    \end{proposition}
   \begin{proof} 
   Given a $W$-strong $J$-relative monad $(T, \eta, (-)^*)$, we define a strong $J$-relative monad ($T$,$\eta$,$(-)^\dagger)$ by
    \begin{equation*}
        (-)^\dagger : \CCat(X\otimes JA, TB)
        \cong \lim_{M\in(W\downarrow X)}\CCat(WM \otimes JA, TB)
        \to \lim_{M\in(W\downarrow X)}\CCat(WM\otimes TA, TB)
        \cong \CCat(X \otimes TA, TB),
    \end{equation*}
    where $(W\downarrow X)$ is a comma category~\cite{mac_lane_categories_1998}. One can check that this provides an inverse to the restriction map defined by $U_W$, and that the strong relative monad laws follow.
\end{proof}
\begin{remark}
    As noted in \Cref{sec:monads-relative-monads}, relative monads are often studied for \emph{well-behaved} roots $J$. In this case, for any $J$-relative monad $T$, the left Kan extension $\Lan_J T$~\cite{mac_lane_categories_1998} exists and receives a monad structure which makes it initial in the comma category $(T\downarrow \res_J)$~\cite{altenkirch_monads_2015}. Under the additional conditions of \Cref{prop:J-strength-is-strong}, with $W=J$, this result extends to \emph{strong} relative monads. In particular, if $T$ is a $J$-strong $J$-relative monad, then $\Lan_J T$ has a strong monad structure which makes it initial in the comma category $(T\downarrow \sres_J)$.
\end{remark}
If $\ACat$ is additionally a symmetric monoidal category and $J$ is a strong symmetric monoidal functor, we can recover the full correspondence of \Cref{prop:equiv-strong-monad} for $J$-strong $J$-relative monads. In this case, we can equivalently express the strength via maps $\theta_{A,B} : JA\otimes TB\to T(A\otimes B)$. This coincides with previous definitions of strong relative monad~\cite{tarmo-grm}, and is also closer to the historical presentation~\cite{kock_strong_1972}.

\begin{definition}\label{def:strength-maps}
    Let $\ACat$ and $\CCat$ be symmetric monoidal categories and $(J,\monmult,\monunit) :  
    \ACat \to \CCat$ a strong symmetric monoidal functor. A \emph{strength map} for a $J$-relative monad $(T, \eta, (-)^*)$ consists of morphisms
    \[
        \theta_{A,B} : JA\otimes TB \to T(A\otimes B),
    \]
    natural in $A,B\in \ob\ACat$, that make the following diagrams commute~\cite{tarmo-grm}.
    \begin{mathpar}
        \begin{tikzcd}[column sep=small]
            {I\otimes TB} & TB \\
            {JI\otimes TA} & {T(I\otimes B)}
            \arrow["{{\lambda_{TB}}}", from=1-1, to=1-2]
            \arrow["{{\monunit\otimes TB}}"', from=1-1, to=2-1]
            \arrow["{{\theta_{I,B}}}"', from=2-1, to=2-2]
            \arrow["{{T\lambda_B}}"', from=2-2, to=1-2]
        \end{tikzcd}
        \and
        \begin{tikzcd}[column sep=scriptsize]
            {(JA\otimes JB)\otimes TC} & {J(A\otimes B)\otimes TC} & {T((A\otimes B)\otimes C)} \\
            {JA\otimes (JB\otimes TC)} & {JA\otimes T(B\otimes C)} & {T(A\otimes (B\otimes C))}
            \arrow["{{\monmult_{A,B}\otimes TC}}", from=1-1, to=1-2]
            \arrow["{{\alpha_{JA,JB,TC}}}"', from=1-1, to=2-1]
            \arrow["{{\theta_{A\otimes B,C}}}", from=1-2, to=1-3]
            \arrow["{{T\alpha_{A,B,C}}}", from=1-3, to=2-3]
            \arrow["{{JA\otimes \theta_{B,C}}}"', from=2-1, to=2-2]
            \arrow["{{\theta_{A,B\otimes C}}}"', from=2-2, to=2-3]
        \end{tikzcd}
        \and
        \begin{tikzcd}
            {JA\otimes JB} & {J(A\otimes B)} \\
            {Ja\otimes Tb} & {T(a\otimes B)}
            \arrow["{\monmult_{A,B}}", from=1-1, to=1-2]
            \arrow["{JA\otimes \eta_B}"', from=1-1, to=2-1]
            \arrow["{\eta_{A\otimes B}}", from=1-2, to=2-2]
            \arrow["{\theta_{A,B}}"', from=2-1, to=2-2]
        \end{tikzcd}
        \and
        \begin{tikzcd}
            {JA\otimes TB} & {T(A\otimes B)} \\
            {JA\otimes TC} & {T(A\otimes C)}
            \arrow["{\theta_{A,B}}", from=1-1, to=1-2]
            \arrow["{JA\otimes f^*}"', from=1-1, to=2-1]
            \arrow["{(\theta_{A,C}\circ (JA\otimes f)\circ \monmult_{A,B}^{-1})^*}", from=1-2, to=2-2]
            \arrow["{\theta_{A,C}}"', from=2-1, to=2-2]
        \end{tikzcd}
    \end{mathpar}
    A morphism $\gamma : (T,\eta,(-)^*) \to (S,\nu, (-)^\dagger)$ of $J$-relative monads is said to \emph{preserve strength} if
    \[
        \gamma_{A\otimes B} \circ \theta_{A,B} = \theta'_{A,B} \circ (JA \otimes \gamma_B)
    \]
\end{definition}

\begin{proposition} \label{prop:J-strength}
    Let $\ACat$ and $\CCat$ be symmetric monoidal categories and $(J,\monmult,\monunit) :  
    \ACat \to \CCat$ a strong symmetric monoidal functor. The following categories are isomorphic:
   \begin{enumerate} 
    \item The category of $J$-strong $J$-relative monads
    \item The category of $J$-relative monads with strength maps and morphisms preserving strength
    \item The category of $(\widehat{\ACat},\widehat{\otimes},\widehat{I})$-enriched $\widehat{J}$-relative monads, where 
    \begin{itemize}
        \item $\ACat$ is given the $\yo$-fully faithful $\widehat{\ACat}$-enrichment
        \item $\CCat$ is enriched in $\widehat{\ACat}$ by $\CCat[X,Y]=\CCat(J{-} \otimes X,Y)$
        \item $\widehat{J}$ is the enriched functor with components 
        \begin{displaymath}
            \ACat[A,B] \cong \ACat({-} \otimes A,B)
            \to \CCat(J({-} \otimes A),JB)
            \cong \CCat(J{-} \otimes JA,JB)
            \cong \CCat[JA,JB].
        \end{displaymath}
    \end{itemize}
    \end{enumerate}
\end{proposition}
\begin{proof}
For the equivalence of $(1)$ and $(2)$, if $(T,\eta,(-)^*)$ is a $J$-strong $J$-relative monad, then
\[
    \theta_{A,B} = (\eta_{A\otimes B} \circ \monmult_{A, B})^*
\]
define strength maps for the underlying relative monad of $(T,\eta,(-)^*)$. (See also the similar result in ~\cite{andrew_slattery_strong_nodate}.)
Conversely, given strength maps, we define
\[
    f^\dagger= (f \circ \monmult^{-1}_{A,B})^* \circ \theta_{A,B}
\]
for any $f: JA\otimes JB \to TC$. One can check that this defines a $J$-strong $J$-relative monad, and that these mappings are mutual inverses.
For the equivalence of $(1)$ and $(3)$, the data provides an enrichment of $\CCat$ and $J$, using the fact that $J$ is strong symmetric monoidal.
The correspondence on units is given by 
\begin{displaymath}
    \widehat{\ACat}(\widehat{I},\CCat[JA,TA]) \cong \widehat{\ACat}(\widehat{I},\CCat(J{-}\otimes JA,TA))
    \cong \CCat(JI\otimes JA,TA)
    \cong \CCat(JA,TA),
\end{displaymath}
and on the extensions by
\begin{displaymath}
    \widehat{\ACat}(\CCat[JA,TB],\CCat[TA,TA]) = \widehat{\ACat}(\CCat(J{-}\otimes JA,TB),\CCat(J{-} \otimes TA,TB)).
\end{displaymath}
The coherence conditions lift through these isomorphisms.
\end{proof}

One motivating example for $J$-strengths is when $J$ is the Yoneda embedding $\yo : \CCat \to \widehat{\CCat}$ of a small category $\CCat$, as in Example \ref{ex:yoneda-relative}. Then a $J$-strong $J$-relative monad is a \textit{strong promonad}~\cite{tarmo-grm} or a notion of \textit{arrow}~\cite{hughes-arrows}. In this case, $J$-strong $J$-relative monads, in fact, coincide with the notion of strong $J$-relative monad in \Cref{def:str-rel-monad}.

\begin{corollary}
    The Yoneda embedding satisfies the conditions of \Cref{prop:J-strength-is-strong}, and so $U_\yo: \mathbf{SMon}(\yo) \to \yo\mathbf{SMon}(\yo)$ is an isomorphism of categories.
\end{corollary}
\section{Encoding graded monads in the linear-non-linear relative monadic metalanguage} \label{sec:gmm}

Graded monads~\cite{katsumata_parametric_2014} are a semantic tool for analysing the usage of effects. They stratify computation via annotation with grades, permitting a more fine-grained analysis of programs. Graded monads have been used for memory and output analysis for the state and writer monad, non-determinism with cuts~\cite{katsumata_flexible_2022, katsumata_parametric_2014}, and to encode epistemic uncertainty and programmable inference in probabilistic programs~\cite{liell-cock_compositional_2025, lew_trace_2020}. Graded monads are also understood to be instances of enriched relative monads~\cite{mcdermott_flexibly_2022}, and in \Cref{thm:strong-graded-as-relative}, we show that strong graded monads can be understood as strong relative monads. We exploit this semantic correspondence to refine RMM to a term calculus for graded monads, based on the linear-non-linear calculus~\cite{benton-lnl}.
\newcommand{\mplus}{\oplus}
\begin{definition}[\cite{katsumata_parametric_2014}] \label{def:graded-monad}
    Let $(\MCat,\mplus,e)$ be a symmetric monoidal category. A \emph{strong $\MCat$-graded monad} on a Cartesian category $(\ACat,\times,1)$ consists of
    \begin{itemize}
        \item Functors $T_{(-)} A : \MCat\op \to \ACat$ for each $A \in \ob{\ACat}$
        \item A unit map $\eta_A : A \to T_e A$ for each $A \in \ob{\ACat}$
        \item A collection of maps $(-)^*_{m,n} : \ACat(\Gamma \times A,T_nB)\to \ACat(\Gamma \times T_m A,T_{m\mplus n} B)$ natural in $\Gamma \in \ob{\ACat}$ and $n \in \ob{\MCat}$ and extranatural in $m\in \ob{\MCat}$
    \end{itemize}
    such that for any $f : \Gamma \times A \to T_m B$ and $g :  \Gamma \times B\to T_n C$ the following equations hold,
    \begin{align*}
        (\eta_A \circ \pi_A)^*_{m,e} &= T_{\rho_m} A \circ \pi_{T_mA}, \\
        (f)^*_{e,m}\circ (\Gamma \times \eta_A) &= T_{\lambda_m} B \circ f, \\
        (g)^*_{l\mplus  m,n} \circ (\pi_\Gamma,(f)^*_{l,m}) &= T_{\alpha_{l,m,n}}C \circ 
        ((g)^*_{m,n} \circ (\pi_\Gamma,f))^*_{l,m \mplus  n},
    \end{align*}
    where $\pi_\Gamma : \Gamma \times A \to \Gamma$ and $\pi_A : \Gamma \times A \to A$ are the projections.
\end{definition}

\begin{example}
    The list monad on $\Set$ is a  strong $(\mathbb{N},\geq,\cdot,1)$ graded monad, where $T_nX=\{\vec{x} \in X^* \mid \mathsf{length}(\vec{x})\leq n \}$.
\end{example}
Monadic style programming can be refined into a graded monadic metalanguage (GMM), by grading the computation types by a symmetric monoidal category $\MCat$~\cite{katsumata_parametric_2014}. The language GMM$(\MCat)$ is defined as follows.
\begin{alignat*}{4}
    &\text{Types:} \quad & A,B &::= 1 ~|~ A\times B ~|~ T_m A & \tag{$m \in \ob{\MCat}$} \\
    &\text{Contexts:} \quad & \Gamma &::= \emptyset ~|~ \Gamma, x:X &
\end{alignat*}
GMM(\MCat) has standard Cartesian term constructors along with the following terms for graded monadic sequencing:
\begin{mathpar}
\inferrule{\tj \Gamma u {A}}{\tj \Gamma {\ret{u}} {T_e A}}
\and
\inferrule{\tj {\Gamma} u {T_m A} \and \tj {\Gamma, x:A} t {T_n B}}{\tj {\Gamma} {\doin{x}{u}{t}} {T_{m\mplus n}B}}
\and
\inferrule{\tj \Gamma {u} {T_n A}}
{\tj \Gamma {\mathbf{T_\xi}A \, u} {T_m A}}(\xi \in \MCat(m,n))
\end{mathpar}
In particular, $\mathsf{return}$ assigns a value at the unit grading, while sequencing computations tensors the grades of the individual computations. These terms satisfy associativity and unitality equations (see \Cref{subsec:GMM-equations} for the full equational theory and term calculus).
\begin{definition}
    A model of GMM(\MCat) consists of 
    \begin{itemize}
        \item A Cartesian category $\CCat$
        \item A strong $\MCat$-graded monad
    \end{itemize}
\end{definition}
The graded monadic metalanguage was introduced in a richer setting with sum types~\cite{katsumata_parametric_2014}, and the models of that type theory have additional structure. However, the same soundness results apply in our restricted setting, and we obtain completeness.
\begin{theorem}[\cite{katsumata_parametric_2014}] \label{thm:completeness-gmm}
    The models of GMM(\MCat) provide a sound and complete semantics.
\end{theorem}

\subsection{Strong graded monads as strong relative monads}

The data of an $\MCat$-graded monad is in bijective correspondence with the data for an $\widehat{\MCat}$-enriched relative monad for a suitably chosen root ~\cite[Thm. 1]{mcdermott_flexibly_2022}. In \Cref{thm:strong-graded-as-relative}, we extend this to strong graded monads. We begin by explicitly describing this connection for $\MCat$-graded monads on $\Set$, which are always strong. In this case, the enrichment of \cite[Thm. 1]{mcdermott_flexibly_2022} is equivalent to strength of the relative monad by \Cref{thm:closed-enriched-strong-corresp}. Hence, we can understand $\MCat$-graded monads on $\Set$ as strong $J$-relative monads, where $J:\Set \to \widehat{\MCat}$ is defined by
\begin{displaymath}
    JX= \MCat(-,e) \cdot X \text,
\end{displaymath}
with $Y\cdot X$ denoting the copower of $X$ by $Y$ and $\widehat{\MCat}$ is considered with its Day monoidal structure. We will now explicitly describe how this corresponds to a strong graded monad. Note that $J$ has a right adjoint, $J\dashv R$, where $R$ is given by evaluation on $e$, since
\begin{displaymath}
    \widehat{\MCat}(JX,F) \cong \Set(X, \widehat{\MCat}(\MCat(-,e), F(-))) \cong \Set(X,Fe).
\end{displaymath}
Hence, the correspondence on the data of the units is via the adjunction:
\begin{displaymath}
    \widehat{\MCat}(JX,TY) \cong \Set(X,T_eY)
\end{displaymath}
Let $W: \MCat \times \Set \to \widehat{\MCat}$ be defined by $W(m,\Gamma)=J\Gamma\otimes \MCat(-,m)$. Then, the following diagram gives a correspondence between extension operators for a strong $\MCat$-graded monad on $\Set$ and extension operators for a $W$-strong $J$-relative monad.
\[\begin{tikzcd}[column sep = small]
	{\Set(\Gamma \times X,T_nY)} & {\widehat{\MCat}(J(\Gamma \times X),\MCat(-,n)\multimap TY)} & {\widehat{\MCat}((J\Gamma \otimes \MCat(-,n))\otimes JX,TY)} \\
	{\widehat{\MCat}(\Gamma \times T_{-}X,T_{-\oplus n}Y)} & {\widehat{\MCat}(J\Gamma \otimes TX,\MCat(-,n)\multimap TY)} & { \widehat{\MCat}((J\Gamma \otimes \MCat(-,n))\otimes TX,TY)}
	\arrow["\cong"{description}, draw=none, from=1-1, to=1-2]
	\arrow["{(=)^*_{-,n}}"', from=1-1, to=2-1]
	\arrow["\cong"{description}, draw=none, from=1-2, to=1-3]
	\arrow["{(=)^*}", from=1-3, to=2-3]
	\arrow["\cong"{description}, draw=none, from=2-1, to=2-2]
	\arrow["\cong"{description}, draw=none, from=2-2, to=2-3]
\end{tikzcd}\]
Note that we have used the natural isomorphisms $Fn \cong R(\MCat(-,n)\multimap F)$ and $(J\Gamma\otimes F)(m)\cong X\times Fm$, which hold for any presheaf $F: \MCat\op \to \Set$, and that $J$ is strong monoidal.
By Proposition \ref{prop:J-strength-is-strong}, this also corresponds to the data for a strong $J$-relative monad. We now generalise this example to any Cartesian closed category with small limits and colimits.

\newcommand{\edaytimes}{\mathbin{\widebar\otimes}}
\begin{theorem} \label{thm:strong-graded-as-relative}
    Let $\CCat$ be a Cartesian closed category with all small limits and colimits, and $\MCat$ a small symmetric monoidal category. Let $\widebar{\MCat} = [\MCat\op,\CCat]$ and define the functor $J : \CCat \to \widebar{\MCat}$ by
    \begin{displaymath}
        JX=\MCat(-,e)\cdot X 
    \end{displaymath}
    Then there is a bijective correspondence between strong $J$-relative monads on $(\widebar{\MCat},\edaytimes,\MCat(-,e)\cdot 1)$ and strong $\MCat$-graded monads on $\CCat$, where $\edaytimes$ is the $\CCat$-enriched Day convolution~\cite{day_closed_1970}:
    \begin{displaymath}
        (F \edaytimes G)(m)= \int^{n,p\in \MCat} \MCat(m,n\mplus  p)\cdot(Fn \times Gp)
    \end{displaymath}
\end{theorem}
\begin{proof}
    The $\CCat$-enriched Day convolution for the free $\CCat$-enrichment of $\MCat$ defines a symmetric monoidal closed structure on $\widebar{\MCat}$. $J$ has a right adjoint $R$ given by evaluation on $e$ and is strong monoidal. The isomorphisms described above for $\Set$ apply in the general case, giving the correspondence between units and the extension operators of a strong graded monad and a $W$-strong $J$-relative monad, where $W : \MCat \times \CCat \to \widebar{\MCat}$ is given by $W(m,X)=\MCat(-,m)\cdot X$. \Cref{prop:J-strength-is-strong} does not directly apply, but its argument can be generalised to give extension operators
    \begin{equation*}
        \widebar{\MCat}(\Gamma \edaytimes JX,TY) \cong \int_n \widebar{\MCat}(W(n,\Gamma n) \edaytimes JX,TY) \to \int_n \widebar{\MCat}(W(n,\Gamma n) \edaytimes TX,TY) \cong \widebar{\MCat}(\Gamma \edaytimes TX,TY),
    \end{equation*}
    since $W$ is enriched dense.
\end{proof}
\newcommand{\aj}[3]{#1 \vdash_{\ACat} #2:#3}
\newcommand{\cj}[4]{#1;#2\vdash_{\CCat} #3:#4}

\subsection{The linear-non-linear RMM}

The linear-non-linear calculus (LNL)~\cite{benton-lnl} is a language that accommodates both linear and non-linear programs. It provides a type system that incorporates \emph{resource sensitivity} via its linear part, with the expressivity of the simply typed lambda calculus, which allows functions such as $\lambda(x:\mathsf{int}).x+x$. \emph{Grades} can be interpreted as cost-sensitive computations, so by this analogy and \Cref{thm:strong-graded-as-relative}, we integrate LNL and RMM to produce a graded monadic metalanguage that is a conservative extension of GMM.
The LNL calculus consists of the following types:
\begin{alignat*}{4}
    &\text{$\ACat$ types:} \quad & A,B &::= 1 ~|~ A\times B ~|~ A\to B ~|~ RX & \\
    &\text{$\CCat$ types:} \quad & X,Y &::= I ~|~ X\otimes Y ~|~ X \multimap Y ~|~ JA & \\
    &\text{$\ACat$ contexts:} \quad & \Gamma,&::= \emptyset ~|~ \Gamma, a:A & \\
    &\text{$\CCat$ contexts:} \quad & \Delta &::= \emptyset ~|~ \Delta, x:X &
\end{alignat*}
There is a \textit{non-linear $\ACat$-judgement} $\aj \Gamma u A$ and a \textit{mixed linear-non-linear $\CCat$-judgement} $\cj \Gamma \Delta t X$. Beyond the terms for the simply typed lambda calculus for $\ACat$ types and terms for the linear lambda calculus for $\CCat$ types (Appendix \ref{app:LNL-RMM}), there are linear-non-linear interaction terms.
\newcommand{\derek}[1]{\mathsf{derelict} \, #1}
\newcommand{\app}[2]{#1 \, #2}
\begin{mathpar}
    \inferrule{\aj \Gamma u A}{\cj \Gamma \emptyset {J(u)} {JA}}
    \and
    \inferrule{\cj \Gamma \emptyset t X}{\aj \Gamma {R(t)} {RX}}
    \\
    \inferrule{\cj \Gamma {\Delta_1} t {JA} \and \cj {\Gamma,a:A} {\Delta_2} s X}{ \cj \Gamma {\Delta_1, \Delta_2} {\letin{J(a)}{t}{s}} X}
    \and
    \inferrule{\aj \Gamma u {RX}}{\cj \Gamma \emptyset {\derek {u}} X}
\end{mathpar}
The following models provide a sound semantics for LNL~\cite{benton-lnl} with the equations in \Cref{app:LNL-RMM}.
\begin{definition} \label{def:LNL-model}
    An LNL model consists of 
    \begin{itemize}
        \item A Cartesian closed category $\ACat$
        \item A symmetric monoidal closed category $\CCat$
        \item A strong symmetric monoidal functor $J : \ACat \to \CCat$ with a right adjoint $J \dashv R$
    \end{itemize}  
\end{definition}
\begin{example} \label{ex:LNL-model}
    If $\CCat$ is a Cartesian closed category with all small limits and colimits, and $\MCat$ is a small symmetric monoidal category, then the triple $(\CCat,\widebar{\MCat},J)$ given in Theorem \ref{thm:strong-graded-as-relative} is an LNL model. The right adjoint $R$ is given by evaluation on $e$.
\end{example}
In view of this example, \Cref{thm:strong-graded-as-relative} shows that we can obtain a graded monad on this LNL model from any strong $J$-relative monad. In fact, this construction applies to any LNL model. We extend LNL accordingly, with a strong relative monad and additional \textit{grade types} $\underline{m}$ which allow the interpretation of the objects of $\MCat$ inside $\CCat$. This defines a new language LNL-RMM($\MCat$) for $\MCat$-graded monads, defined as follows.
\begin{alignat*}{4}
    &\text{$\ACat$ types:} \quad & A,B &::= 1 ~|~ A\times B ~|~ A\to B ~|~ RX & \\
    &\text{$\CCat$ types:} \quad & X,Y &::= I ~|~ \underline{m} ~|~ X\otimes Y ~|~ X \multimap Y ~|~ JA ~|~ TA  \tag{$m \in \ob{\MCat}$} &
\end{alignat*}
The linear RMM terms are given by:
\begin{mathpar}
    \inferrule{\cj \Gamma \Delta  t {JA}}{\cj \Gamma \Delta {\ret{t}} {TA}}
    \and
    \inferrule{\cj \Gamma {\Delta_1} t {TA} \and \cj \Gamma {\Delta_2, x:JA} s {TB}}{\cj \Gamma {\Delta_1,\Delta_2} {\doin{x}{t}{s}} {TB}}
\end{mathpar}
The embedding of the grade types is given by:
\newcommand{\merg}[1]{{\mathsf{merge} \ #1}}
\newcommand{\unmerg}[1]{{\mathsf{unmerge} \ #1}}
\begin{mathpar}
    \inferrule{\cj \Gamma {\Delta}  t {\underline{m}}}{\cj \Gamma \Delta {\underline{\xi} \, t} {\underline{n}}}(\xi \in \MCat(m,n))
    \and 
    \inferrule{\cj \Gamma \Delta  t {\underline{e}}}{\cj \Gamma \Delta {\unmerg t} {I}}
    \and
    \inferrule{\cj \Gamma \Delta  t I}{\cj \Gamma \Delta {\merg t} {\underline{e}}}
    \and
    \inferrule{\cj \Gamma \Delta  t {\underline{m\mplus n}}}{\cj \Gamma \Delta {\unmerg t} {\underline{m}\otimes \underline{n}}}
    \and 
   \inferrule{\cj \Gamma \Delta  t {\underline{m}\otimes \underline{n}}}{\cj \Gamma \Delta {\merg t} {\underline{m\mplus n}}}
\end{mathpar}
The full term calculus and its equations are given in \Cref{app:LNL-RMM}. We can now define a strong $\MCat$-graded monad on the $\ACat$-types in the term calculus of LNL-RMM(\MCat). The graded monad type constructor $T_m A$ on $\mathcal A$ is given by:
\begin{equation}\label{eqn:GMasRM}
    T_mA=R(\underline{m} \multimap TA)
\end{equation}
The unit of the graded monad is then: 
\[
\lambda (a:A).\ R(\lambda (x:\underline{e}).\letin{\lunit}{\unmerg x}{\ret{J(a)}}) : A\to T_e A,
\]
The bind $\bind \ : T_mA\to (A\to T_nB) \to T_{m\mplus n}B$ is:
\begin{multline*}
  \lambda (f:T_mA).\ \lambda (g:A\to T_nB).\ 
R\big(\lambda (s:\underline{m \oplus n}).\letin{(t,r)}{\unmerg s} \\{\big(\doin{x}{\app{(\derek{f}){t}}}{(\letin{J(a)}{x}{\app{(\derek{\app{g}{a}})}{r}}})}\big)\big)
\end{multline*}
For $\xi : m \to n$, the regrading morphism $T_\xi X : T_n X \to T_mX$ is given by:
\begin{displaymath}
    \lambda(f:T_n X).R\big(\lambda(s:\underline{m}).\app{(\derek{f})}(\app {\underline{\xi}} {s})\big)
\end{displaymath}
\begin{example}
    If a symmetric monoidal category $\MCat$ is \emph{finitely generated}, then we can simplify the signature of LNL-RMM($\MCat$) by only adding generator grade types, and encoding the relations in the equational theory. This makes LNL-RMM($\MCat$) suitable for implementations of a graded monad. For example, an $(\mathbb{N},\geq,+,0)$-graded monad can be specified by the following type signature.
\begin{alignat*}{4}
    &\text{$\ACat$ types:} \quad & A,B &::= 1 ~|~ A\times B ~|~ A\to B ~|~ RX & \\
    &\text{$\CCat$ types:} \quad & X,Y &::= I ~|~ \underline{1} ~|~ X\otimes Y ~|~ X \multimap Y ~|~ JA ~|~ TA &
\end{alignat*}
The additional terms and equations required are:
\begin{mathpar}
    \inferrule{ }{\cj \Gamma \emptyset * \underline{1}}
    \and
    \inferrule{\cj \Gamma \emptyset s {\underline{m}\multimap \underline{n}} \and \cj \Gamma \emptyset t {\underline{m}\multimap \underline{n}}}{\cj \Gamma \emptyset {s=t} {\underline{m}\multimap \underline{n}}}
\end{mathpar}
where $\underline{m}$ is syntactic sugar for the $m$-fold product $\underline{1}\otimes\dots\otimes\underline{1}$.
\end{example}

We can illustrate the utility of the linear and non-linear aspects of the language for graded monads via two example programs. 

\begin{example} The non-linearity of the term calculus for $\ACat$ types allows us to define a term that self-sequences an effectful function, which would not be possible in a purely linear language.
\begin{equation*}
\lambda(a:A).\lambda(f:A\to T_mA).(f \, a \bind f): A\to T_{m\oplus m} A
\end{equation*}
\end{example}
\begin{example}
    The linearity of the term calculus for $\CCat$-types provides a suitable constraint for reasoning about grade types. Linear type theories restrict which terms are derivable by preventing the discarding and duplication of variables. Consider a program in LNL-RMM,
\begin{displaymath}
    \mathsf{runBoth} : T_mA\times T_mB\to T_{m\oplus m}(A\times B) \text,
\end{displaymath}
which runs each computation of cost $m$, returning the paired result necessarily of cost $m\mplus m$.
If we extended the $\CCat$ term derivations to a non-linear type theory, allowing the discarding and duplication of variables, then we could duplicate $\underline{m}$, now permitting the following program:
\begin{displaymath}
    \mathsf{duplicateAndRunBoth}: T_mA\times T_m B\to T_m(A\times B)
\end{displaymath}
Hence, the linearity restriction on derivable terms for $\CCat$-types facilitates effect-usage-sensitive computations.
\end{example}
The strong graded monad on $\ACat$-types defines an embedding of GMM(\MCat) into the term calculus of LNL-RMM(\MCat). This means that there is a map on GMM(\MCat) types and terms (denoted $[-]$) into the $\ACat$-types and derivable terms of LNL-RMM(\MCat) that preserves the equational theory of GMM(\MCat). We will show that this embedding is \emph{conservative}, so that it additionally \emph{reflects} the equational theory of GMM(\MCat):
\[\tj \Gamma {u=v} A \iff \aj {[\Gamma]} {[u]=[v]} {[A]}\]
We begin by identifying the models of LNL-RMM(\MCat).
\begin{definition}
    A model of LNL-RMM(\MCat) is an LNL model $(\ACat,\CCat,J\dashv R)$ along with:
    \begin{itemize}
        \item A strong symmetric monoidal functor $\MCat \to \CCat$ which we denote by $m \mapsto \underline{m}$
        \item A strong $J$-relative monad $(T,\eta,(-)^*)$ 
    \end{itemize}
\end{definition}
\begin{theorem} \label{thm:completeness-lnl-rmm}
    The models of LNL-RMM(\MCat) are a sound semantics.
\end{theorem}
\begin{proof}
    The construction of the type theory is modular; i.e. we add terms from linear RMM with no additional equations with the terms of LNL. A model of LNL-RMM(\MCat) is therefore a model of LNL and (linear) RMM. One can also directly verify that the additional terms and equations for the grade types $\MCat$ are validated by a strong symmetric monoidal functor.
\end{proof}
In categorical semantics, it is common to show completeness of a class of models by constructing a term model from the types, terms and equations. We are not aware of a proof of completeness of LNL in this fashion for the models in \Cref{def:LNL-model}, and that result is not within the scope of this paper. It is worth noting that due to the modularity of our type theory, a completeness result for LNL-RMM(\MCat) would follow from the completeness of LNL. However, the soundness of LNL-RMM(\MCat) (\Cref{thm:completeness-lnl-rmm}) combined with the completeness of GMM(\MCat) (\Cref{thm:completeness-gmm}), is sufficient for us to show the conservativity of the interpretation of graded monads in LNL-RMM(\MCat).

\begin{theorem} \label{thm:conservative-gmm}
    The interpretation of GMM(\MCat) types and terms in the term calculus of LNL-RMM(\MCat) is a conservative extension.
\end{theorem}
\begin{proof}
        Suppose $(\ACat,T)$ is a small model of GMM(\MCat) (where $T$ denotes a strong $\MCat$-graded monad). A strong $\MCat$-graded monad on a small Cartesian category $\ACat$ can equivalently be given by 
        \begin{itemize}
           \item Endofunctors $T_m : \CCat \to \CCat$
           \item Maps $(\theta_m)_{A,B} : A\otimes T_mB \to T_m(A\otimes B)$ natural in $A,B \in \ob{\ACat}$
           \item Natural transformations $T_\xi : T_n \to T_m$ for each $\xi : m \to n$
           \item A natural transformation $\eta : \id_\ACat \to T_e$
           \item Natural transformations $\mu_{m,n} : T_m T_n \to T_{m\oplus n}$
        \end{itemize}
        satisfying equations~\cite{katsumata_parametric_2014}. By applying the 2-functor $\ACat \mapsto \widehat{\ACat}$ to this data, we obtain a strong $\MCat$-graded monad on $\CCat=\widehat{\ACat}$. We thus have a model $(\CCat,\widehat{T})$ of GMM($\MCat$) and a full and faithful functor $\yo : \ACat \to \CCat$ which strictly preserves the interpretation of GMM($\MCat$) types and terms. Now, since $\CCat$ is a Cartesian closed category with all small limits and colimits, by \Cref{ex:LNL-model}, the adjunction described in \Cref{thm:strong-graded-as-relative} provides an LNL model. Furthermore, since $\widehat{T}$ is a strong graded monad on $\CCat$, by the correspondence in \Cref{thm:strong-graded-as-relative}, we obtain a model of LNL-RMM(\MCat), where the strong symmetric monoidal functor $\MCat\to \widebar{\MCat}$ is given by the enriched Yoneda embedding $\underline{m}=\MCat(m,-)\cdot 1$. One can also verify that the correspondence in \Cref{thm:strong-graded-as-relative} yields (up to isomorphism) the strong graded monad on $\CCat$, defined by $\den{R(m\multimap TX)}$. Hence, any small model of GMM embeds fully faithfully into a model of LNL-RMM(\MCat), and so by letting $\ACat$ be the term model of GMM(\MCat), we see that the embedding is a conservative extension.
    \end{proof}

\section{The arrow calculus in the relative monadic metalanguage} \label{sec:armm}

\newcommand{\fst}[1]{\mathsf{fst} \ #1}
\newcommand{\snd}[1]{\mathsf{snd} \ #1}
\newcommand{\bang}{\mathbin{!}}
\newcommand{\aabs}[3]{{\lambda(#1:#2).#3}}
\newcommand{\aapp}[2]{{#1 \bullet #2}}
\newcommand{\fapp}[2]{{#1 \ #2}}
\newcommand{\extract}[4]{{\mathsf{extract}(#1#2.#3, #4)}}
\newcommand{\K}[3]{{K(#1.#2,#3)}}
\newcommand{\J}[5]{{J(#1#2.#3, #4, #5)}}

Arrows are programming constructs that generalise monads to computations that are not strictly sequential~\cite{hughes-arrows}. A calculus for arrows and corresponding semantics has been developed over the years~\cite{hughes-arrows, Paterson2001, ATKEY201119, SANADA_2024}, culminating in \emph{the arrow calculus}~\cite{arrow-calc}, which has a strong resemblance to the RMM framework introduced in this paper, with analogous unit and bind terms. On the semantic side, this is not surprising since arrows have been described as strong monads relative to the Yoneda embedding (see \cite{tarmo-grm} and \Cref{sec:j-strengths}).
 The arrow calculus extends the simply typed lambda calculus with a binary type constructor $\leadsto$.
\begin{alignat*}{4}
    &\text{Types:} \quad & A,B &::= 1 ~|~ A\times B ~|~ A\to B ~|~ A \leadsto B & \\
    &\text{Contexts:} \quad & \Gamma,\Delta &::= \emptyset ~|~ \Gamma, x : A &
\end{alignat*}
Beyond the standard typing derivations of the simply typed lambda calculus, there is a new \emph{command} typing judgement with two contexts:
\[
    \Gamma; \Delta \vdash t \bang A
\]
The additional command terms are given by an arrow abstraction and application, and a unit and bind for arrow types.
\begin{mathpar}
    \inferrule{
    \Gamma; x:A \vdash t \bang B
    }{
    \Gamma \vdash \lambda^\bullet x.t : A \leadsto B
    }
    \and
    \inferrule{
    \Gamma \vdash u : A \leadsto B
    \and
    \Gamma, \Delta \vdash v : A
    }{
    \Gamma; \Delta \vdash u \bullet v \bang B
    }
    \\
    \inferrule{
    \Gamma, \Delta \vdash u : A
    }{
    \Gamma; \Delta \vdash \ret{u} \bang A
    }
    \and
    \inferrule{
    \Gamma; \Delta \vdash t \bang A
    \and
    \Gamma; \Delta, x:A \vdash s \bang B
    }{
    \Gamma; \Delta \vdash {\doin{x}{t}{s}} \bang B
    }
\end{mathpar}
Note the adjusted syntax for the return and bind terms for consistency with the languages of this paper. There are equations for beta and eta equality of the arrow abstraction, and the monad laws.
\begin{mathparpagebreakable}    
    \inferrule{\Gamma; x:A \vdash t \bang B \and \Gamma,\Delta \vdash u : A}{\Gamma; \Delta \vdash \lambda^\bullet x.t \bullet u \ = \ t[ u/ x] \bang B}
    \and
    \inferrule{\Gamma \vdash u : A \leadsto B}{\Gamma \vdash \lambda^\bullet x.(u \bullet x) \ = \ u : A \leadsto B}
    \and
    \inferrule{\Gamma, \Delta \vdash u : A
    \and
    \Gamma; \Delta, x:A \vdash t \bang B}
    {\Gamma; \Delta \vdash \doin{x}{\ret{u}}{t} \ = \ t[u / x] \bang B}
    \and
    \inferrule{\Gamma; \Delta \vdash t \bang B}{\Gamma; \Delta \vdash \doin{x}{t}{\ret{x}} \ = \ t \bang B}
    \and
    \inferrule{\Gamma; \Delta \vdash r \bang A
    \and
    \Gamma; \Delta, x:A \vdash s \bang B
    \and
    \Gamma; \Delta, y:B \vdash t \bang C}{
    \Gamma; \Delta \vdash \doin{x}{r}{(\doin{y}{s}{t})} \ = \ \doin{y}{(\doin{x}{r}{s})}{t} \bang C}
\end{mathparpagebreakable}
\newcommand{\AAA}{{\ACat\Cat(\ACat\op, \ACat)}}

Motivated by the similarities to the RMM, we provide a sound and complete denotational semantics for the arrow calculus using a refined notion of strong relative monads (\Cref{sec:j-strengths}).

\begin{definition} \label{def:arrow-model}
    An arrow calculus model consists of
    \begin{itemize}
        \item A Cartesian closed category $\ACat$
        \item A $W$-strong $J$-relative monad $(T,\eta, (-)^*)$
    \end{itemize}
    where $J$ and $W$ are defined as follows. If $\ACat$ is a Cartesian closed category, we write $\AAA$ for the category of $\ACat$-enriched functors $\ACat\op \to \ACat$ and enriched natural transformations~\cite{kelly}. Then $J,K: \ACat \to \AAA$ are defined by the \emph{enriched Yoneda embedding}, $JAB = A^{B}$, and the constant preseaf, $KAB = A$, respectively. Finally, $W: \ACat\times \ACat \to \AAA$ is defined as $W(A,B) = KA \times JB$ and intuitively captures the two contexts in the arrow calculus command judgements.
\end{definition}

We interpret judgements $\tj \Gamma u A$ as morphisms in $\ACat$ and the terms of the simply typed lambda calculus by their standard interpretation in the Cartesian closed category $\ACat$~\cite{jacobs_categorical_2001, lambek_introduction_1986}. The new type is interpreted as
\begin{equation}
    \den{A \leadsto B} = T\den{B} \den{A}\text, \label{eq:command-j}
\end{equation}
and the new typing judgement is interpreted as a morphism in $\ACat$. \[ \den{\Gamma; \Delta \vdash s \bang B} :\den{\Gamma} \to T\den{B}\den{\Delta} \]
The interpretation of the remaining terms is defined inductively.
\begin{align*}
    \den{\Gamma \vdash \lambda^\bullet x.t : A \leadsto B} &= \den{\Gamma; x:A \vdash t \bang B} \\
    \den{\Gamma; \Delta \vdash u \bullet v \bang B} &= \mathsf{ev}_{T\den{B}\den{\Delta},T\den{B}\den{A}} \circ (T\den{B}_{\den{A},\den{\Delta}}
    \circ \mathsf{curry}_{\den{\Gamma},\den{\Delta}, \den{A}} (\den{\Gamma, \Delta \vdash v: A}), \\
    &\qquad\quad \den{\Gamma \vdash u : A \leadsto B}) \\
    \den{\Gamma; \Delta \vdash \ret{u} \bang A} &= \chi_{\den{\Gamma},\den{\Delta},T\den{A}}(\eta_A \circ \chi_{\den{\Gamma}, \den{\Delta},J\den{A}}^{-1} (\mathsf{curry}_{\den{\Gamma},\den{\Delta},\den{A}} (\den{\Gamma, \Delta\vdash u : A}))) \\
    \den{\Gamma; \Delta \vdash \doin{x}{t}{s} \bang B} &= \chi_{\den{\Gamma},\den{\Delta},T{\den{B}}}\Big(\Big(\chi_{\den{\Gamma},\den{\Delta}\times{\den{A}}, T\den{B}}^{-1}(\den{\Gamma; \Delta, x:A \vdash s \bang B}) \\
    & \qquad  \circ (\id_{K\den{\Gamma}} \times \kappa_{\den{\Delta}, \den{A}})\Big)^* \circ \big(\id_{K\den{\Gamma} \times J \den{\Delta}}, \chi_{\den{\Gamma},\den{\Delta}, T\den{A}}^{-1}(\den{\Gamma; \Delta \vdash t \bang A})\big)\Big)
\end{align*}
Here, $\mathsf{ev}_{A,B}: A^B \times B \to A$ is the evaluation morphism given by the closed structure, $\mathsf{curry}_{A,B,C}: \ACat(A\times B, C) \to \ACat(A, C^B)$ is the closed structure adjunction, $TA_{B,C} : B^C \to TAC^{TAB}$ is the enriched functor components, $\kappa_{A,B} \colon JA\times JB\to J(A\times B)$ is the canonical map, and $\chi_{A,B,F}: \AAA(KA \times JB, F)  \xrightarrow{\sim}
 \ACat(A, FB)$ is the adjunction given by $\theta \mapsto \lambda a. \theta_B(a, \lambda b.b)$ in the forward direction and $f \mapsto \theta$, where $\theta_C = \mathrm{ev}_{FC,FB} \circ (F_{B,C} \circ \pi_{B^C}, f\circ \pi_A) : A \times B^C \to FC$, in the reverse direction.
\begin{theorem} \label{thm:completeness-arrow}
    The models of the arrow calculus are a sound and complete semantics.
\end{theorem}
\begin{proof}
Soundness is deduced from the following substitution lemmas, which are proved by structural induction.
\begin{align*}
    \sem{\Gamma \vdash u[v / b] : B} &= \sem{\Gamma, a:A \vdash u : B} \circ (\id_{\sem{\Gamma}}, \sem{\Gamma \vdash v : A}) \\
    \sem{\Gamma; \Delta \vdash t[u / a] \bang B} &= \sem{\Gamma, a : A; \Delta \vdash t \bang B} \circ (\id_{\sem{\Gamma}}, \sem{\Gamma \vdash u : A}) \\
    \sem{\Gamma; \Delta \vdash t[u / a] \bang B} &= \chi_{\sem{\Gamma}, \sem{\Delta}, T\sem{B}}\Big( \chi^{-1}_{\sem{\Gamma}, \sem{\Delta}\times \sem{A}, T\sem{B}}(\sem{\Gamma; \Delta, a:A\vdash t \bang B}) \\
    &\qquad \circ (\pi_{K\sem{\Gamma}}, \kappa_{\sem\Delta, \sem A} \circ (\pi_{J\sem{\Delta}}, \chi^{-1}_{\sem{\Gamma}, \sem{\Delta}, J\sem{A}} (\mathsf{curry}_{\den{\Gamma},\den{\Delta},\den{A}} (\sem{\Gamma, \Delta \vdash u : A})))) \Big)
\end{align*}
Completeness is proved via the term model construction. The simply typed lambda calculus terms give the Cartesian closed structure \cite{jacobs_categorical_2001, lambek_introduction_1986}. The object mapping for the $W$-strong $J$-relative monad is $({-}\leadsto B)$, where the enriched functorial action is given by $\lambda^\bullet a.\aapp{t}{\app{f}{a}}$ for $f : A'\to A$ and $t: A\leadsto B$. Next, every equivalence class of terms $[\Gamma; \Delta \vdash t \bang B]$ corresponds under the adjunction $\chi$ to a natural map $K\sem\Gamma \times J\sem\Delta \to ({-}\leadsto \sem B)$. The equivalence classes of $\ret{t}$ and $\doin{x}{p}{q}$ terms in contexts $\Gamma; \Delta$ under the adjunction then define the $W$-strong $J$-relative monad structure.
\end{proof}

\newcommand{\coerce}[1]{{\mathsf{coerce} \ #1}}
\subsection{The arrow relative monadic metalanguage (ARMM)} \label{sec:armm-rmm}

Identifying the models of the arrow calculus as a $W$-strong relative monad suggests that it can be extended to the framework of RMM. We begin by identifying the key features of the models in \Cref{def:arrow-model} to achieve this. The command judgement~\eqref{eq:command-j}, while interpreted as a morphism in $\ACat$, is primarily used in the semantics as a morphism in $\AAA$. The critical link is via the following isomorphism:
\begin{equation}
    \AAA(KA\times JB, TC)\cong \ACat(A, TCB). \label{eq:aaa-iso}
\end{equation}

The objects in $\AAA$ can be thought of as arrow types with parametric input. In \Cref{def:arrow-model}, the semantics of the types $JA$ is given by the enriched Yoneda embedding, and so can be understood as pure arrow types with output $A$. Types $KA$ have semantics via constant presheaves and do not interact with the relative monad structure. The role of these types is via \eqref{eq:aaa-iso}; they permit the movement of variables between $\ACat$ and $\AAA$ contexts. Finally, types $TA$ are effectful arrow types with output $A$, parametric over some input.

However, the role of $\AAA$ is not essential apart from providing canonical interpretations of $J$ and $K$. We can therefore abstract the explicit category $\AAA$ to an arbitrary Cartesian category $\CCat$, while maintaining Cartesian functors $J, K: \ACat \to \CCat$ satisfying certain conditions. In particular, we require that parametric arrow types in $\CCat$ can be abstracted over a pure input (capturing the isomorphism \eqref{eq:aaa-iso}), and abstracting pure arrows results in function types:
\begin{align}
    \CCat(KA \times JB, X) &\cong \ACat(A, B\Rightarrow X) \label{eqn:arr-abs} \\
    \ACat(A \times B, C) &\cong \ACat(A, B\Rightarrow JC) \label{eqn:fun-abs}
\end{align}
Then, we can define the original arrow types as $A \leadsto B ::= A \Rightarrow TB$ and function types as $A \to B ::= A \Rightarrow JB$. We additionally require a natural coercion from $K$ to $J$ to allow variables in the $J$ context to be moved to variables in the $K$ context. This is admissible in the arrow calculus and is semantically validated by the canonical map from the Yoneda embedding to the constant presheaf functor. It is derivable from \eqref{eqn:arr-abs} and \eqref{eqn:fun-abs} by applying the following isomorphism to the identity morphism.
\begin{displaymath}
    \ACat(A, A) \cong \ACat(A \times 1, A) \cong \ACat(A, 1 \Rightarrow JA) \cong \CCat(KA \times J1, JA) \cong \CCat(KA, JA)
\end{displaymath}
We thus define the arrow relative monadic metalanguage (ARMM) as follows.
\begin{alignat*}{4}
    &\text{$\ACat$ types:} \quad & A,B &::= 1 ~|~ A\times B ~|~ A\Rightarrow X & \\
    &\text{$\CCat$ types:} \quad & X,Y &::= 1 ~|~ X\times Y ~|~ JA ~|~ KA ~|~ TA & \\
    &\text{$\ACat$ contexts:} \quad & \Gamma, \Delta &::= \emptyset ~|~ \Gamma, a:A & \\
    &\text{$\CCat$ contexts:} \quad & \Phi &::= \emptyset ~|~ \Phi, x:X &
\end{alignat*}
\newcommand{\ej}[5]{#1;#2;#3 \vdash_{\CCat} #4:#5}

The ARMM terms for $\ACat$-types are the standard terms for a Cartesian category. The ARMM terms for $\CCat$-types beyond the standard Cartesian ones are:
\begin{mathpar}
    \inferrule{\ej \Gamma \Delta \Phi {t} {KA} 
    \and
    \ej {\Gamma, a:A} {\Delta} {\Phi} {s} {X}}{
    \ej \Gamma \Delta {\Phi} {\letin{K(a)}{t}{s}} X
    }
    \and
    \inferrule{\ej \Gamma \Delta \Phi {t} {JA} 
    \and
    \ej {\Gamma} {\Delta, a:A} {\Phi} {s} {X}}{
   \ej \Gamma \Delta {\Phi} {\letin{J(a)}{t}{s}} X
    }
    \and
    \inferrule{\ej \Gamma \Delta \Phi t {JA}}{\ej \Gamma \Delta \Phi {\ret{t}} {TA}}
    \and
    \inferrule{\ej \Gamma \Delta {\Phi} t {TA} 
    \and
    \ej {\Gamma} {\Delta} {x:JA} s {TB}
    }{
    \ej \Gamma \Delta \Phi {\doin{x}{t}{s}} {TB}
    }
\end{mathpar}
Note that $\CCat$ judgements have three contexts. The first two typing rules validate that context $\Gamma$ is a $K$ context and context $\Delta$ is a $J$ context, similar to the style of the arrow calculus. The third context $\Phi$ can have arbitrary $\CCat$ types, such as computations, and permits categorical semantics via the general RMM framework. The latter two typing rules are the standard $J$-relative monad terms for $T$, where binding is restricted to $K$ and $J$ contexts. The ARMM terms for the interaction between $\ACat$-types and $\CCat$-types are:
\begin{mathpar}
    \inferrule{\aj \Gamma u {A\Rightarrow X} \and \aj {\Gamma,\Delta} {v} A}{\ej \Gamma \Delta \Phi {\aapp u v} {X}}
    \and
    \inferrule{\ej \Gamma {a:A} \emptyset t X}{\aj \Gamma {\aabs a {A} t} {A\Rightarrow X}}
    \\
    \inferrule{\aj {\Gamma} u A}{\ej \Gamma \Delta \Phi {K(u)} {KA}}
    \and
    \inferrule{\aj {\Gamma,\Delta} u A}{\ej {\Gamma} {\Delta} \Phi {J(u)} {JA}}
    \and
    \inferrule{\aj {\Gamma} u {A \Rightarrow JB} \and
    \aj \Gamma v A}{\aj {\Gamma} {\fapp{u}{v}} {B}}
\end{mathpar}
The full term calculus and equational theory of ARMM is given in \Cref{app:armm-equations}.
Regular application and abstraction under the derivable function type $A \to B ::= A \Rightarrow JB$ are:
\begin{mathpar}
    \inferrule{\aj \Gamma u {A \to B} \and \aj \Gamma v A}{\aj \Gamma {\fapp{u}{v}} B}
    \and
    \inferrule{\aj {\Gamma, a:A} u {B}}{\aj \Gamma {\aabs{a}{A}{J(u)}} {A \to B}}
\end{mathpar}

\begin{definition}A model of ARMM consists of 
\begin{itemize}
    \item Cartesian categories $\ACat$ and $\CCat$
    \item Cartesian functors $J,K : \ACat \to\ \CCat$
    \item A $W$-strong $J$-relative monad $(T,\eta,(-)^*)$, where $W(A,B) = KA \times JB$.
    \item An adjunction $(K{-}\times JA) \dashv (A \Rightarrow {-})$ for each $A \in \ob{\ACat}$.
    \item An adjunction $({-}\times A) \dashv (A \Rightarrow J{-})$ for each $A \in \ob{\ACat}$.
\end{itemize}
\end{definition}

The $\ACat$ judgements are interpreted by morphisms in $\ACat$ as usual, while the $\CCat$ judgements are given by morphisms in $\CCat$:
    \begin{displaymath}
        \den{\ej \Gamma \Delta \Phi t X} : K\den{\Gamma} \times J\den{\Delta} \times \den{\Phi} \to \den{X}
    \end{displaymath}
    The non-standard interaction terms are interpreted inductively as:
    \begin{align*}
        \den{\ej \Gamma \Delta \Phi {\aapp{u}{v}} X} &= \chi_{\den{\Gamma},\den{A}, \den{X}}^{-1}(\den{\aj \Gamma u {A\Rightarrow X}}) \\
        & \qquad  \circ (\pi_{K\den{\Gamma}},J\den{\aj {\Gamma,\Delta} {v} A} \circ \kappa^J_{\den{\Gamma},\den{\Delta}} \circ (\psi_{\den{\Gamma}} \circ \pi_{K\den{\Gamma}}, \pi_{J\den{\Delta}})) \\
        \den{\aj \Gamma {\aabs{x}{A}{t}} A \Rightarrow X} &= \chi_{\den{\Gamma}, \den{A}, \den{X}}(\den{\ej \Gamma {x:A} \emptyset t X}) \\
        \den{\ej \Gamma \Delta \Phi {K(u)} {KA}} &= K\den{\aj \Gamma u A} \circ \pi_{K\den{\Gamma}} \\
        \den{\ej \Gamma \Delta \Phi {J(u)} {JA}} &= J\den{\aj {\Gamma,\Delta} u A} \circ \kappa^J_{\den{\Gamma},\den{\Delta}} \circ (\psi_{\den{\Gamma}} \circ \pi_{K\den{\Gamma}}, \pi_{J\den{\Delta}}) \\
        \den{\aj \Gamma {\fapp{u}{v}} B} &= \mathsf{curry}_{\den{\Gamma},\den{A}, \den{B}}^{-1}(\den{\aj \Gamma u {A \Rightarrow JB}}) \circ (\id_\den{\Gamma}, \den{\aj \Gamma v A})
    \end{align*}
    where $\chi_{A,B,X}: \CCat(KA\times JB, X) \to \ACat(A, B\Rightarrow X)$ and $\mathsf{curry}_{A,B,C}: \ACat(A\times B, C) \to \ACat(A, B \Rightarrow JC)$ are the components of the adjunctions, respectively, $\psi_A : KA \to JA$ are the components of the natural coercion, and $\kappa^J : JA\times JB \to J(A\times B)$ and $\kappa^K : KA\times KB \to K(A\times B)$ are the canonical isomorphisms.

\begin{example} \label{ex:armm-model}
    If $\ACat$ is a Cartesian closed category and $T : \ACat \to \ACat\Cat(\ACat\op,\ACat)$ is a $W$-strong $J$-relative monad, then $(\ACat,J,K,T)$ is a model of ARMM where $J$ and $K$ are as in \Cref{def:arrow-model}.
\end{example}
\begin{theorem} \label{thm:completeness-armm}
    The models of ARMM are a sound semantics.
\end{theorem}

\begin{proof}
Soundness is shown via structural induction on the following substitution lemmas.
\begin{align*}
    \sem{\Gamma \vdash u[v / b] : B} &= \sem{\Gamma, a:A \vdash u : B} \circ (\id_{\sem{\Gamma}}, \sem{\Gamma \vdash v : A}) \\
    \sem{\ej \Gamma \Delta \Phi {t[u / a]} X} &= \sem{\ej {\Gamma, a : A} \Delta \Phi t X} \\
    &\qquad \circ (\kappa^K_{\sem\Gamma, \sem A} \circ (\pi_{K\sem\Gamma}, K\sem{\aj \Gamma u A} \circ \pi_{K\sem\Gamma}), \pi_{J\sem\Delta}, \pi_{\sem{\Phi}}) \\
    \sem{\ej \Gamma \Delta \Phi {t[u / a]} X} &= \sem{\ej {\Gamma} {\Delta, a : A} \Phi t X} \circ (\pi_{K\sem\Gamma}, \\
    &\qquad \kappa^J_{\sem\Delta, \sem A} \circ (\pi_{J\sem\Delta}, J\sem{\aj {\Gamma,\Delta} u A} \circ (\psi_{\sem\Gamma} \circ \pi_{K\sem\Gamma}, \pi_{J\sem\Delta})), \pi_{\sem\Phi}) \\
    \sem{\ej \Gamma \Delta \Phi {t[s / y]} X} &= \sem{\ej \Gamma \Delta \Phi {t[s / y]} X} \circ (\id_{K\sem\Gamma \times J\sem\Delta \times \sem\Phi}, \sem{\ej \Gamma \Delta \Phi s Y})
\end{align*}
The equations in \Cref{app:armm-equations} are also validated by structural induction. The adjunction $\chi$ gives the $\beta$ and $\eta$ laws for application and abstraction of arrows.
\end{proof}

With this result, we can show that our calculus subsumes the arrow calculus. The syntactic interpretation of their types is $\sden{A \to B} = (A \Rightarrow JB)$ and $\sden{A \leadsto B} = (A \Rightarrow TB)$. The command judgement \eqref{eq:command-j} is interpreted as
\begin{displaymath}
    \sden{\Gamma;\Delta \vdash t \bang A} \quad = \quad \ej \Gamma \Delta \emptyset t {TA}
\end{displaymath}
while their terms are interpreted as:
\begin{align*}
    \sden{\Gamma \vdash \lambda^\bullet x.t : A \leadsto B} &\quad=\quad \aj \Gamma {\aabs{x}{A}{t}} {A\Rightarrow TB} \\
    \sden{\Gamma; \Delta \vdash u\bullet v \bang B} &\quad=\quad \ej \Gamma \Delta \emptyset {\aapp{u}{v}} {TB} \\
    \sden{\Gamma; \Delta \vdash \ret{u} \bang A} &\quad=\quad \ej \Gamma \Delta \emptyset {\ret{J(u)}} {TA} \\
    \sden{\Gamma; \Delta \vdash \doin{x}{t}{s} \bang B} &\quad=\quad \ej \Gamma \Delta \emptyset {\doin{y}{t}{(\letin{J(x)}{y}{s}})} {TB}
\end{align*}

\begin{theorem} \label{thm:conservative-armm}
    The interpretation of the arrow calculus in ARMM is a conservative extension.
\end{theorem}
\begin{proof}
    From \Cref{ex:armm-model}, all models of the arrow calculus $(\ACat, T)$ are models of ARMM. Furthermore, the semantics of the interpretation of an arrow calculus term coincides, or corresponds via $\chi$, with its semantics as a model of the arrow calculus. Hence, the term model must validate all the equations of ARMM.
\end{proof}

\begin{remark}
The ARMM calculus could be generalised further by allowing binds in arbitrary contexts.
\begin{mathpar}
    \inferrule{\ej \Gamma \Delta {\Phi} t {TA} 
    \and
    \ej {\Gamma} {\Delta} {\Phi, x:JA} s {TB}
    }{
    \ej \Gamma \Delta \Phi {\doin{x}{t}{s}} {TB}
    }
\end{mathpar}
This would make the models consist of a strong $J$-relative monad, rather than a $W$-strong $J$ relative monad. We cannot apply \Cref{prop:J-strength-is-strong} to equate these notions of strength because $\ACat$ is not cocomplete. However, an analogous conservative extension result may be possible by embedding~$\ACat$ into a cocomplete category, such as $\widehat{\ACat}$.
\end{remark}

\section{Conclusion and future work}

We have introduced a language RMM for strong relative monads and demonstrated a significant advantage to working with term calculi based on it. RMM generalises the monadic metalanguage~\cite{moggi_notions_1991} in two directions:
\begin{itemize}
    \item It allows us to \textit{restrict} our notion of computation by programming with $\CCat$ types that only allow computations on another sort of types $\ACat$. This was done in \Cref{sec:finitary} to restrict to computations on \textit{finite} datatypes and thus provide a more fine-grained approach to computation.
    \item It allows us to \textit{expand} our notion of computation. By programming with $\ACat$ types with \textit{partial access} to an additional sort $\CCat$, it permits a broader notion of computations than what is allowed in $\ACat$. We saw in \Cref{sec:gmm} and \Cref{sec:armm} that this allowed programming with cost-sensitive computations through graded monads and a wider variety of effects through arrows.
\end{itemize}
In particular, the languages LNL-RMM introduced in \Cref{sec:gmm} for graded monads and ARMM in \Cref{sec:armm} for arrows:
\begin{enumerate}
    \item Are no less general than their predecessors since we prove conservative extension theorems, \Cref{thm:conservative-gmm} and \Cref{thm:conservative-armm}.
    \item Are more expressive than their predecessors, where, in particular, LNL-RMM has grade types in the language, which allows for more flexible programming.
    \item Come with built-in sound denotational semantics from the RMM, as shown in \Cref{thm:completeness-lnl-rmm} and \Cref{thm:completeness-armm}.
\end{enumerate}
The denotational semantics of all the example languages in the paper, including the new semantics provided for the arrow calculus (\Cref{thm:completeness-arrow}), was possible due to the notion of strong relative monad we introduced in \Cref{sec:strong-strong}, and its generalisations in \Cref{sec:j-strengths}. We justified that this definition of strong relative monads is the appropriate one for semantics of programming languages by establishing its equivalence to enriched relative monads (\Cref{thm:hat-strong-corresp}). We also extended the broader literature on relative monads~\cite{arkor_formal_2024, arkor_nerve_2024, arkor_relative_2025, tarmo-grm} to facilitate the design of these languages.

\paragraph{Future work.}
We have focused on two applications of relative monads: graded monads and arrows.
The positive results suggest that other applications of relative monads will also benefit from taking the relative monadic metalanguage as a starting point. One example is the \textit{simple relational specification monads}~\cite{maillard_next_2020}, which are (strong) relative monads on the product functor $\Set \times \Set \to \Set$. 

Since we have shown that the relative monadic metalanguage is a unifying starting point for
diverse aspects of programming languages, we envisage that more advanced methods from monad theory will benefit from a general analysis in terms of relative monads. We have in mind methods from program logics, such as monad liftings for logical relations~\cite{katsumata_semantic_2005,shin-ya_katsumata_codensity_2018}, which have not yet been investigated from the relative point of view. 

\paragraph{Acknowledgements.}
This work was funded in part by ERC Grant BLAST, AFOSR under award number FA9550-21-1-0038, grants from ARIA Safeguarded AI, and the Clarendon Scholarship.

We thank the anonymous reviewers and many people for helpful discussions and suggestions, including Dylan McDermott, Nathanael Arkor and Nathan Corbyn.

\bibliographystyle{ACM-Reference-Format}
\bibliography{refs}

\appendix

\section{Bistrong relative monads on non-symmetric monoidal categories} \label{app:bistrength}

If $\CCat$ is a non-symmetric monoidal category, \Cref{def:str-rel-monad} no longer allows sequencing of computations inside arbitrary contexts. It permits a bind for terms with contexts of the form $\Gamma,x:JA$, but not the generic form of $\Gamma,x:JA,\Gamma'$. In this more general setting, the notion described in \Cref{def:str-rel-monad} is referred to as a \emph{left strong} monad. A \emph{right strong} monad can be defined analogously, and formally understood as a left strong monad on the \emph{reverse monoidal category}.
\begin{definition}
    If $\CCat$ is a monoidal category, its reverse $\CCat\rev$ is defined to have the same underlying category, but $X\otimes\rev Y= Y\otimes X$, $\rho\rev=\lambda$, $\lambda\rev=\rho$ and $\alpha\rev=\alpha^{-1}$
\end{definition}
 In the non-relative case, when $\CCat$ is not a symmetric monoidal category, it is common to consider a \emph{bistrong} monad on $\CCat$. Usually, a bistrength for a monad is presented as a compatible left and right strength~\cite{dylan-tarmo-strong}. However, as we have seen, strength maps are insufficient for relative monads. The Kleisli form of a bistrong relative monad provides the more suitable notion. The following is an instantiation of a \textit{multicategorical strong relative monad}~\cite{slattery-thesis} for a representable multicategory~\cite{hermida_representable_2000}:

\begin{definition} \label{def:bistrong-relative}
    A \emph{bistrong $J$-relative monad} $(T, \eta, (-)^*)$, consists of:
    \begin{itemize}
        \item An object $TA\in \ob{\CCat}$ for each $A \in \ob{\ACat}$,
        \item A morphism $\eta_A: A \to TA$ for each $A \in \ob{\ACat}$,
        \item A collection of maps $(-)^*: \CCat(\Gamma \otimes (JA\otimes \Delta),TB)\to \CCat(\Gamma\otimes (TA\otimes \Delta),TB)$ natural in $\Gamma,\Delta\in \CCat$
    \end{itemize}
    all such that for $f: \Gamma\otimes (JA\otimes \Delta) \to TB$ and $g: \Gamma'\otimes (JB\otimes \Delta') \to TC$, the following equations hold:
    \begin{align*}
        (\eta_A \circ \rho_{JA} \circ \lambda_{JA\otimes I})^* &= \rho_{TA}\circ \lambda_{TA \otimes I} \\
        f^*\circ (\Gamma\otimes (\eta_A\otimes \Delta)) &= f \\
        g^* \circ (\Gamma'\otimes (f^*\otimes \Delta')) \circ \widetilde{\alpha}_{\Gamma',\Gamma,TA,\Delta,\Delta'} &= (g^* \circ (\Gamma'\otimes (f\otimes \Delta')) \circ \widetilde{\alpha}_{\Gamma',\Gamma,JA,\Delta,\Delta'})^*
    \end{align*}
    where $\widetilde{\alpha}_{\Gamma',\Gamma,X,\Delta,\Delta'} = (\Gamma'\otimes (\alpha^{-1}_{\Gamma,X\otimes \Delta,\Delta'}\circ (\Gamma\otimes \alpha^{-1}_{X,\Delta,\Delta'}))) \circ \alpha_{\Gamma',\Gamma,X\otimes (\Delta\otimes \Delta')}$.
    
    When $\CCat$ is symmetric with symmetry $\sigma$, a bistrong $J$-relative monad is additionally said to be \textit{symmetric} when for $f:(\Gamma\otimes \Delta)\otimes (JA\otimes \Gamma') \to TB$,
    \[
        f^* \circ \widetilde\sigma_{\Gamma,TA,\Delta,\Gamma'} = (f \circ \widetilde\sigma_{\Gamma,JA,\Delta,\Gamma'})^*
    \]
    where $\widetilde{\sigma}_{\Gamma,X,\Delta,\Gamma'} : \Gamma\otimes (X\otimes (\Delta\otimes \Gamma')) \to (\Gamma\otimes \Delta) \otimes (X\otimes \Gamma')$ is given by
    \[
        \widetilde{\sigma}_{\Gamma,X,\Delta,\Gamma'} = \alpha^{-1}_{\Gamma,\Delta,{X\otimes \Gamma'}}\circ(\Gamma\otimes (\alpha_{\Delta,X,\Gamma'}\circ (\sigma_{X,\Delta}\otimes \Gamma') \circ\alpha^{-1}_{X,\Delta,\Gamma'})).
    \]

A \textit{morphism of (symmetrically) bistrong $J$-relative monads} $ \gamma : (T,\eta,(-)^*) \to (S,\nu, (-)^\dagger)$ consists of a family of maps $\gamma_A : TA \to SA$ satisfying, for any $f : \Gamma\otimes (JA\otimes \Delta)  \to TB$,
\begin{mathpar}
    \nu_A = \gamma_A \circ \eta_A
    \and
    \gamma_B \circ f^* = (\gamma_B \circ f)^\dagger \circ  (\Gamma\otimes (\gamma_A\otimes \Delta))
\end{mathpar}
We now write $\mathbf{SMon}_B(J)$, $\mathbf{SMon}_L(J)$, and $\mathbf{SMon}_R(J)$ for the categories of bistrong, left strong, and right strong $J$-relative monads and morphisms, respectively.
\end{definition}
There are forgetful functors: 
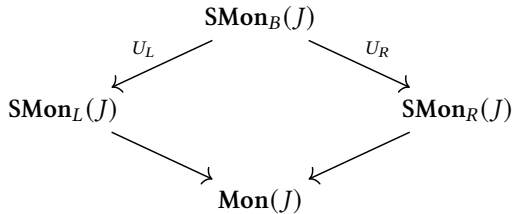
\[\begin{tikzcd}
	& {\mathbf{SMon}_B(J)} \\
	{\mathbf{SMon}_L(J)} && {\mathbf{SMon}_R(J)} \\
	& {\mathbf{Mon}(J)}
	\arrow["{U_L}"', from=1-2, to=2-1]
	\arrow["{U_R}", from=1-2, to=2-3]
	\arrow[from=2-1, to=3-2]
	\arrow[from=2-3, to=3-2]
\end{tikzcd}\]
which commute with the \textit{underlying $J$-relative monad} functors. When $\CCat$ is symmetric monoidal, and we restrict $\mathbf{SMon}_B(J)$ to the full subcategory of symmetrically bistrong $J$-relative monads, $U_L$ and $U_R$ become isomorphisms.
\begin{remark}
    For a jointly left and right strong $J$-relative monad $(T,\eta,(-)^{*_L},(-)^{*_R})$ (such that the underlying $J$-relative monads coincide), we can formulate a compatibility condition for any $f:\Gamma\otimes JA \to JB$, $g:JA\otimes \Delta \to JC$, $h:JB\otimes \Delta \to TD$, and $k: \Gamma\otimes JC \to TD$:
    \[\begin{tikzcd}[column sep=2.25em]
        {(\Gamma\otimes JA)\otimes \Delta} & {\Gamma\otimes (JA\otimes \Delta)} \\
        & {\Gamma\otimes JC} \\
        {JB\otimes \Delta} & TD
        \arrow["{\alpha_{\Gamma,JA,\Delta}}", from=1-1, to=1-2]
        \arrow["{{{f\otimes \Delta}}}"', from=1-1, to=3-1]
        \arrow["{{{\Gamma\otimes g}}}", from=1-2, to=2-2]
        \arrow["k", from=2-2, to=3-2]
        \arrow["h"', from=3-1, to=3-2]
    \end{tikzcd} \implies \begin{tikzcd}[column sep=2.25em]
        {(\Gamma\otimes TA)\otimes \Delta} & {\Gamma\otimes (TA\otimes \Delta)} \\
        & {\Gamma\otimes TC} \\
        {TB\otimes \Delta} & TD
        \arrow["{\alpha_{\Gamma,TA,\Delta}}", from=1-1, to=1-2]
        \arrow["{{{(\eta_B \circ f)^{*_L}\otimes \Delta}}}"', from=1-1, to=3-1]
        \arrow["{{{\Gamma\otimes (\eta_C \circ g)^{*_R}}}}", from=1-2, to=2-2]
        \arrow["{{{k^{*_L}}}}", from=2-2, to=3-2]
        \arrow["{{{h^{*_R}}}}"', from=3-1, to=3-2]
    \end{tikzcd}\]
    Any bistrong $J$-relative monad satisfies this condition, and when $J=\id_\CCat$, this recovers the definition of a bistrong monad, but this is not generally sufficient to define the structure of a bistrong $J$-relative monad. For linear type-theoretic applications, \Cref{def:bistrong-relative} is more suitable, as it allows computations to be sequenced anywhere in a context.
\end{remark}

The notion of $J$-strength can also be extended to $J$-bistrength in a canonical way. We note this definition because \Cref{prop:J-strength-is-strong} holds at the generality of (symmetric) bistrong monads and $J$-relative monads with a (symmetric) $J$-bistrength.

\begin{definition}
Let $\ACat$ and $\CCat$ be monoidal categories and $J \colon \ACat\to \CCat$ a strong monoidal functor. Given a $J$-relative monad $(T,\eta,(-)^*)$, a left $J$-strength $\theta^L$ (\Cref{def:strength-maps}), and right $J$-strength $\theta^R$ (i.e. a left $J$-strength on $\CCat\rev$), we say $(\theta^L,\theta^R)$ is a \textit{$J$-bistrength} for $T$ if the following diagram commutes.
   \[\begin{tikzcd}
	{(JA\otimes TB)\otimes JC} & {T(A\otimes B)\otimes JC} & {T((A\otimes B)\otimes C)} \\
	{JA\otimes (TB\otimes JC)} & {JA\otimes T(B\otimes C)} & {T(A\otimes (B\otimes C))}
	\arrow["{{\theta^L_{A,B}\otimes JC}}", from=1-1, to=1-2]
	\arrow["{{\alpha_{JA,TB,JC}}}"', from=1-1, to=2-1]
	\arrow["{{\theta^R_{A\otimes B,C}}}", from=1-2, to=1-3]
	\arrow["{{T\alpha_{A,B,C}}}", from=1-3, to=2-3]
	\arrow["{{JA\otimes \theta^R_{B,C}}}"', from=2-1, to=2-2]
	\arrow["{{\theta^L_{A,B\otimes C}}}"', from=2-2, to=2-3]
\end{tikzcd}\]
If $\ACat$ and $\CCat$ are symmetric monoidal categories and $J$ is strong symmetric monoidal, we say $T$ has \textit{symmetric $J$-bistrength} if the following diagram commutes.
\[\begin{tikzcd}
	{JA\otimes TB} & {T(A\otimes B)} \\
	{TB\otimes JA} & {T(B\otimes A)}
	\arrow["{\theta^L_{A,B}}", from=1-1, to=1-2]
	\arrow["{\sigma_{JA,TB}}"', from=1-1, to=2-1]
	\arrow["{T\sigma_{A,B}}", from=1-2, to=2-2]
	\arrow["{\theta^R_{A,B}}"', from=2-1, to=2-2]
\end{tikzcd}\]
\end{definition}

\section{Full term calculi and their equational theories} \label{app:eqns}

We require the following reflexivity, symmetry, transitivity, and substitution laws to hold for all judgements for the remainder of this appendix:
\begin{mathpar}
    \inferrule{\tj \Gamma u X}{\tj \Gamma {u=u} X}
    \and
    \inferrule{\tj \Gamma {u=v} X}{\tj \Gamma {v=u} X}
    \and
    \inferrule{\tj \Gamma {u=v} X \and \tj \Gamma {v=w} X}{\tj \Gamma {u=w} X}
    \and
    \inferrule{\tj \Gamma u X \and \tj {\Gamma,x:X} {s=t} Y}{\tj \Gamma {s[u/x]=t[u/x]} Y}
\end{mathpar}
We also require that all term constructors satisfy \textit{congruence} equations, e.g. for product types:
\begin{mathpar}
  \inferrule{\tj \Gamma {u=v} {X\times Y}}{\tj \Gamma {\pi_1 u=\pi_1 v} X}
    \and
    \inferrule{\tj \Gamma {u=v} {X\times Y}}{\tj \Gamma {\pi_2 u=\pi_2 v} Y}
    \and
    \inferrule{\tj \Gamma {u=v} {X} \and \tj \Gamma {t=s} {Y}}{\tj \Gamma {(u,t)=(v,s)} {X\times Y}}  
\end{mathpar}

\subsection{Relative monadic metalanguage (RMM)} \label{app:RMM-equations}

The types and contexts of RMM are:
\begin{alignat*}{3}
   &\text{Types:} \quad &X,Y &::= 1 ~|~ X\times Y ~|~  JA ~|~ TA  \tag{$A\in \ob\ACat$} \\
   &\text{Contexts:} \quad &\Gamma &::= \emptyset ~|~ \Gamma,x:X 
\end{alignat*}
The terms are:
\begin{mathpar}
    \inferrule{ }{\tj {\Gamma,x:X} x {X}}
    \and
    \inferrule{ }{\tj {\Gamma} \lunit {1}}
    \and
    \inferrule{\tj \Gamma u {X\times Y}}{\tj \Gamma {\pi_1 u} {X}}
    \and
    \inferrule{\tj \Gamma u {X\times Y}}{\tj \Gamma {\pi_2 u} {Y}}
    \and
    \inferrule{\tj {\Gamma} u {X} \and \tj {\Gamma} t {Y}}{\tj {\Gamma} {(u,t)} {X\times Y}}
    \and
    \inferrule{\tj \Gamma {u} {JA}}{\tj \Gamma {\mathbf{f} \, u} {JB}}(f \in \ACat(A,B))
    \and
    \inferrule{\tj \Gamma u {JA}}{\tj \Gamma {\ret{u}} {TA}}
    \and
    \inferrule{\tj {\Gamma} u {TA} \and \tj {\Gamma, y:JA} t {TB}}{\tj {\Gamma} {\doin{y}{u}{t}} {TB}}
\end{mathpar}
The additional equations beyond the symmetry, reflexivity, transitivity, congruence and substitution laws are:
\begin{mathparpagebreakable}
     \inferrule{\tj \Gamma u 1 }{\tj \Gamma {u= \lunit} 1}
    \and
    \inferrule{\tj \Gamma {u} {X} \and \tj \Gamma {t} {Y}}{\tj \Gamma {\pi_1 (u,t)=u} {X}}
    \and
    \inferrule{\tj \Gamma {u} {X} \and \tj \Gamma {t} {Y}}{\tj \Gamma {\pi_2 (u,t)=t} {Y}}
    \and
    \inferrule{\tj \Gamma {u} {X\times Y}}{\tj \Gamma {(\pi_1 u,\pi_2 u)=u} {X\times Y}}
    \and
    \inferrule{\tj \Gamma {u} {JA}}{\tj \Gamma {\mathbf{id_A} \; u = u} {JA}}
    \and
    \inferrule{\tj \Gamma {u} {JA}}{\tj \Gamma {\mathbf{(g\circ f)} \; u = \mathbf{g} \; ( \mathbf{f} \; u)} {JC}}(f \in \ACat(A,B),g\in \ACat(B,C))
    \and
    \inferrule{\tj \Gamma u {JA}
\and \tj {\Gamma,y:JA} t {TB}}{\tj \Gamma {\doin{y}{\ret{u}}{t} = t[u/y]} {TB}}
    \and
    \inferrule{\tj \Gamma u {TA}
}{\tj \Gamma {\doin{y}{u}{\ret y} = u} {TA}}
    \and
    \inferrule{\tj \Gamma u {TA} \and \tj {\Gamma,y:JA} t {TB} \and \tj {\Gamma,z:JB} v {TC}}{\tj \Gamma {\doin{z}{(\doin{y}{u}{t})}{v}= \doin{y}{u}{(\doin{z}{t}{v})}} {TC}}
\end{mathparpagebreakable}
\subsection{Graded monadic metalanguage (GMM)} \label{subsec:GMM-equations}

The types and contexts of GMM are:

\begin{alignat*}{3}
   &\text{Types:} \quad &X,Y &::= 1 ~|~ X\times Y ~|~ T_m X  \tag{$m \in \ob{\MCat}$} \\
   &\text{Contexts:} \quad &\Gamma &::= \emptyset ~|~ \Gamma,x:X 
\end{alignat*}
The terms are:
\begin{mathparpagebreakable}
    \inferrule{ }{\tj {\Gamma,x:X} x {X}}
    \and
    \inferrule{ }{\tj {\Gamma} \lunit {1}}
    \and
    \inferrule{\tj \Gamma u {X\times Y}}{\tj \Gamma {\pi_1 u} {X}}
    \and
    \inferrule{\tj \Gamma u {X\times Y}}{\tj \Gamma {\pi_2 u} {Y}}
    \and
    \inferrule{\tj {\Gamma} u {X} \and \tj {\Gamma} t {Y}}{\tj {\Gamma} {(u,t)} {X\times Y}}
    \and
    \inferrule{\tj \Gamma u X}{\tj \Gamma {\ret{u}} {T_e X}}
    \and
    \inferrule{\tj {\Gamma} u {T_m X} \and \tj {\Gamma, x:X} t {T_n Y}}{\tj {\Gamma} {\doin{x}{u}{t}} {T_{m\mplus n}Y}}
    \and
    \inferrule{\tj \Gamma {u} {T_n X}}{\tj \Gamma {\mathbf{T_\xi}X \, u} {T_m X}}( \xi \in \MCat(m,n))
\end{mathparpagebreakable}
The additional equations beyond the symmetry, reflexivity, transitivity, congruence and substitution laws are:
\begin{mathparpagebreakable}
    \inferrule{\tj \Gamma u 1 }{\tj \Gamma {u= \lunit} 1}
    \and
    \inferrule{\tj \Gamma {u} {X} \and \tj \Gamma {t} {Y}}{\tj \Gamma {\pi_1 (u,t)=u} {X}}
    \and
    \inferrule{\tj \Gamma {u} {X} \and \tj \Gamma {t} {Y}}{\tj \Gamma {\pi_2 (u,t)=t} {Y}}
    \and
    \inferrule{\tj \Gamma {u} {X\times Y}}{\tj \Gamma {(\pi_1 u,\pi_2 u)=u} {X\times Y}}
    \and
    \inferrule{\tj \Gamma {u} {T_m X}}{\tj \Gamma {\mathbf{T}_{\id_m}X \; u = u} {T_m X}}
    \and
    \inferrule{\tj \Gamma {u} {T_l X}}{\tj \Gamma {\mathbf{T_{\varphi \circ \xi}}X \; u = \mathbf{T_\xi}X\; ( \mathbf{T_\varphi} X \; u)} {T_m} X}( \xi \in \MCat(m,n), \varphi \in \MCat(n,l)) 
    \and
    \inferrule{\tj {\Gamma} u {T_l X} \and \tj {\Gamma, x:X} t {T_n Y}}{\tj {\Gamma} {{\doin{x}{u}{(\mathbf{T_\xi} Y \;t)}= \mathbf{T}_{l\mplus \xi} Y \; (\doin{x}{u}{t}})} {T_{l\mplus m}Y}}( \xi \in \MCat(m,n))
    \and
     \inferrule{\tj {\Gamma} u {T_n X} \and \tj {\Gamma, x:X} t {T_l Y}}{\tj {\Gamma} {{\doin{x}{(\mathbf{T_\xi} X \;u)}{t}= \mathbf{T}_{\xi\mplus l} Y \; (\doin{x}{u}{t}})} {T_{m\mplus l}Y}}( \xi \in \MCat(m,n))
     \and
    \inferrule{\tj \Gamma u X 
\and \tj {\Gamma,x:X} t {T_m Y}}{\tj \Gamma {\doin{x}{\ret{u}}{t} =\mathbf{T}_{\lambda_m} Y \; t[u/x]} {T_{e\mplus m}Y}}
\and
\inferrule{\tj \Gamma u {T_m X}
}{\tj \Gamma {\doin{x}{u}{\ret x} = \mathbf{T}_{\rho_m} X \; u} {T_{m\mplus e}X}}
 \and
    \inferrule{\tj \Gamma u {T_l X} \and \tj {\Gamma,x:X} t {T_m Y} \and \tj {\Gamma,y:Y} v {T_n Z}}{\tj \Gamma {\doin{y}{(\doin{x}{u}{t})}{v}= \mathbf{T}_{\alpha_{l,m,n}} Z (\doin{x}{u}{(\doin{y}{t}{v}))}} {T_{(l\mplus m)\mplus n}Z}}
    \end{mathparpagebreakable}
    
\subsection{Linear-non-linear RMM (LNL-RMM)} \label{app:LNL-RMM}

The types and contexts of LNL-RMM are:
 \begin{alignat*}{4}
    &\text{$\ACat$ types:} \quad & A,B &::= 1 ~|~ A\times B ~|~ A\to B ~|~ RX & \\
    &\text{$\CCat$ types:} \quad & X,Y &::= I ~|~ \underline{m} ~|~ X\otimes Y ~|~ X \multimap Y ~|~ JA ~|~ TA  \tag{$m \in \ob{\MCat}$} & \\
    &\text{$\ACat$ contexts:} \quad & \Gamma,&::= \emptyset ~|~ \Gamma, a:A & \\
    &\text{$\CCat$ contexts:} \quad & \Delta &::= \emptyset ~|~ \Delta, x:X &
\end{alignat*}
The non-linear terms for $\ACat$ types are:
\begin{mathparpagebreakable}
    \inferrule{ }{\aj {\Gamma,a:A} a {A}}
    \and
    \inferrule{ }{\aj {\Gamma} \lunit {1}}
    \and
    \inferrule{\aj \Gamma u {A\times B}}{\tj \Gamma {\pi_1 u} {A}}
    \and
    \inferrule{\aj \Gamma u {A\times B}}{\tj \Gamma {\pi_2 u} {B}}
    \and
    \inferrule{\aj {\Gamma} u {A} \and \tj {\Gamma} v {B}}{\tj {\Gamma} {(u,v)} {A\times B}}
    \and
    \inferrule{\aj {\Gamma,a:A} u B}{\aj \Gamma {\lambda (a:A).u} {A\to B}}
    \and
    \inferrule{\aj \Gamma u {A\to B} \and \aj \Gamma v A}{\aj \Gamma {\app u v} B}
\end{mathparpagebreakable}
The linear terms for $\CCat$ types are:
\begin{mathparpagebreakable}
    \inferrule{ }{\cj {\Gamma} {x:X} x X}
    \and
    \inferrule{ }{\cj {\Gamma} \emptyset \lunit {I}}
    \and
    \inferrule{\cj \Gamma {\Delta_1} t X \and \cj \Gamma {\Delta_2} s Y }{\cj \Gamma {\Delta_1,\Delta_2} {(t,s)} {X\otimes Y}}
    \and
    \inferrule{ \cj \Gamma {\Delta_1} t I \and \cj \Gamma {\Delta_2} s X}{ \cj \Gamma {\Delta_1,\Delta_2} {\letin{\lunit}{t}{s}} {X}}
    \and
    \inferrule{\cj \Gamma {\Delta_1} t {X\otimes Y} \and \cj \Gamma {\Delta_2,x:X,y:Y} s Z}{\cj \Gamma {\Delta_1,\Delta_2} {\letin{(x,y)}{t}{s}} {Z}}
    \and
    \inferrule{\cj \Gamma {\Delta,x:X} s Y}{\cj \Gamma \Delta {\lambda (x:X).s} {X\multimap Y}}
    \and
    \inferrule{\cj \Gamma {\Delta_1} s {X\multimap Y} \and \cj \Gamma {\Delta_2}  t X}{\cj \Gamma {\Delta_1,\Delta_2} {\app s t} {Y}}
    \and
    \inferrule{\cj \Gamma \Delta  t {JA}}{\cj \Gamma \Delta {\ret{t}} {TA}}
    \and
    \inferrule{\cj \Gamma {\Delta_1} t {TA} \and \cj \Gamma {\Delta_2, y:JA} s {TB}}{\cj \Gamma {\Delta_1,\Delta_2} {\doin{y}{t}{s}} {TB}}
    \and
    \inferrule{\cj \Gamma {\Delta}  t {\underline{m}}}{\cj \Gamma \Delta {\underline{\xi} \, t} {\underline{n}}}(\xi \in \MCat(m,n))
    \and 
    \inferrule{\cj \Gamma \Delta  t {\underline{e}}}{\cj \Gamma \Delta {\unmerg t} {I}}
    \and
    \inferrule{\cj \Gamma \Delta  t I}{\cj \Gamma \Delta {\merg t} {\underline{e}}}
    \and
    \inferrule{\cj \Gamma \Delta  t {\underline{m\mplus n}}}{\cj \Gamma \Delta {\unmerg t} {\underline{m}\otimes \underline{n}}}
    \and 
   \inferrule{\cj \Gamma \Delta  t {\underline{m}\otimes \underline{n}}}{\cj \Gamma \Delta {\merg t} {\underline{m\mplus n}}}
\end{mathparpagebreakable}
The terms for linear-non-linear interaction are:
\begin{mathpar}
    \inferrule{\aj \Gamma u A}{\cj \Gamma \emptyset {J(u)} {JA}}
    \and
    \inferrule{\cj \Gamma {\Delta_1} t {JA} \and \cj {\Gamma,a:A} {\Delta_2} s X}{ \cj \Gamma {\Delta_1, \Delta_2} {\letin{J(a)}{t}{s}} X}
    \and
    \inferrule{\cj \Gamma \emptyset t X}{\aj \Gamma {R(t)} {RX}}
    \and
    \inferrule{\aj \Gamma u {RX}}{\cj \Gamma \emptyset {\derek {u}} X}
\end{mathpar}
The equations for terms in $\ACat$ are:
\begin{mathparpagebreakable}
    \inferrule{\aj \Gamma u 1 }{\aj \Gamma {u= \lunit} 1}
    \and
    \inferrule{\aj \Gamma {u} {A} \and \aj \Gamma {v} {B}}{\aj \Gamma {\pi_1 (u,v)=u} {A}}
    \and
    \inferrule{\aj \Gamma {u} {A} \and \aj \Gamma {v} {B}}{\aj \Gamma {\pi_2 (u,v)=t} {B}}
    \and
    \inferrule{\aj \Gamma {u} {A\times B}}{\aj \Gamma {(\pi_1 u,\pi_2 u)=u} {A\times B}}
    \and
    \inferrule{\aj {\Gamma} {u} {A} \and \aj {\Gamma,a:A} v B}{\aj \Gamma {\app{(\lambda(a:A).v)}{u}=v[u/a]}
     {B}}  
    \and
    \inferrule{\aj \Gamma {u} {A\to B}}{\aj \Gamma {\lambda(a:A).\app{u}{a}=u} {A\to B}}
\end{mathparpagebreakable}
The equations for the linear lambda calculus terms in $\CCat$ are:
\begin{mathparpagebreakable}
\inferrule{\cj \Gamma {\Delta} t {X}}{\cj \Gamma {\Delta} {\letin{\lunit}{\lunit}{t}=t} {X}}
 \and
 \inferrule{\cj \Gamma {\Delta} t {I}}{\cj \Gamma {\Delta} {\letin{\lunit}{t}{\lunit}=t} {I}}
\and
 \inferrule{\cj \Gamma {\Delta_1} t {X} \and \cj \Gamma {\Delta_2} s Y \and \cj \Gamma {\Delta_3,x:X,y:Y} r Z}{\cj \Gamma {\Delta_1,\Delta_2,\Delta_3} {\letin{(x,y)}{(t,s)}{r}=r[t/x,s/y]} {Z}}
 \and
  \inferrule{\cj \Gamma {\Delta} t {X\otimes Y}}{\cj \Gamma {\Delta} {\letin{(x,y)}{(x,y)}{t}=t} {X\otimes Y}}
 \and
\inferrule{\cj \Gamma {\Delta_1,x:X} s Y \and \cj \Gamma {\Delta_2} t X}{\cj \Gamma {\Delta_1,\Delta_2} {\app{(\lambda (x:X).s)}{t}=s[t/x]} {Y}}
\and
\inferrule{\cj \Gamma {\Delta} t {X\multimap Y}}{ \cj \Gamma \Delta {\lambda(x:X).\app{t}{x}=t} {X\multimap Y}}
\and 
\inferrule{\cj \Gamma {\Delta} t {X\multimap Y}}{ \cj \Gamma \Delta {\lambda(x:X).t x=t} {X\multimap Y}}
\and
\inferrule{\cj \Gamma {\Delta_1} t {W\multimap Z} \and \cj \Gamma {\Delta_2} s {I} \and \cj \Gamma {\Delta_3} r W }{\cj \Gamma {\Delta_1,\Delta_2,\Delta_3} {{\app{(\letin{\lunit}{s}{t})}{r}}=\letin{\lunit}{s}{(\app {t} {r})} } Z}
\and
\inferrule{\cj \Gamma {\Delta_1} t {W\multimap Z} \and \cj \Gamma {\Delta_2} s {I} \and \cj \Gamma {\Delta_3} r W }{\cj \Gamma {\Delta_1,\Delta_2,\Delta_3} {{\app{t}{(\letin{\lunit}{s}{r})}}=\letin{\lunit}{s}{(\app {t} {r})} } Z}
\and
\inferrule{\cj \Gamma {\Delta_1,x:X,y:Y} t {W\multimap Z} \and \cj \Gamma {\Delta_2} s {X\otimes Y} \and \cj \Gamma {\Delta_3} r W }{\cj \Gamma {\Delta_1,\Delta_2,\Delta_3} {{\app{(\letin{(x,y)}{s}{t})}{r}}=\letin{(x,y)}{s}{(\app {t} {r})} } Z}
\and
\inferrule{\cj \Gamma {\Delta_1} t {W\multimap Z} \and \cj \Gamma {\Delta_2} s {X\otimes Y} \and \cj \Gamma {\Delta_3,x:X,y:Y} r W }{\cj \Gamma {\Delta_1,\Delta_2,\Delta_3} {{\app{t}{(\letin{(x,y)}{s}{r})}}=\letin{(x,y)}{s}{(\app {t} {r})} } Z}
\and
\inferrule{\cj \Gamma {\Delta_1,x:X} t {Y} \and \cj \Gamma {\Delta_2} s {I}}{\cj \Gamma {\Delta_1,\Delta_2} {{\lambda(x:X).{(\letin{\lunit}{s}{t})}}=\letin{\lunit}{s}{(\lambda (x:X). t)} } X\multimap Y}
\and
\inferrule{\cj \Gamma {\Delta_1, w:W, z:Z, x:X} t {Y} \and \cj \Gamma {\Delta_2} s {W\otimes Z}}{\cj \Gamma {\Delta_1,\Delta_2} {{\lambda(x:X).{(\letin{(w,z)}{s}{t})}}=\letin{(w,z)}{s}{(\lambda (x:X). t)} } X\multimap Y}
\and
\inferrule{\cj \Gamma {\Delta_1,x:X,y:Y,w:W} t {Z} \and \cj \Gamma {\Delta_2} s {X\otimes Y}}{\cj \Gamma {\Delta_1,\Delta_2,\Delta_3} {{\app{t}{(\letin{(x,y)}{s}{r})}}=\letin{(x,y)}{s}{(\app {t} {r})} } Z}
\and
\inferrule{\cj \Gamma {\Delta_1} t {I} \and \cj \Gamma {\Delta_2} s {I} \and \cj \Gamma {\Delta_3} r X }{\cj \Gamma {\Delta_1,\Delta_2,\Delta_3} {\letin{\lunit}{(\letin{\lunit}{t}{s}))}{r}=\letin{\lunit}{t}{(\letin{\lunit}{s}{r})}} X}
\and
\inferrule{\cj \Gamma {\Delta_1} t {I} \and \cj \Gamma {\Delta_2} s {X\otimes Y} \and \cj \Gamma {\Delta_3,x:X,y:Y} r Z }{\cj \Gamma {\Delta_1,\Delta_2,\Delta_3} {\letin{(x,y)}{(\letin{\lunit}{t}{s}))}{r}=\letin{\lunit}{t}{(\letin{(x,y)}{s}{r})}} Z}
\and
\inferrule{\cj \Gamma {\Delta_1,x:X,y:Y} t {I} \and \cj \Gamma {\Delta_2} s {X\otimes Y} \and \cj \Gamma {\Delta_3} r Z }{\cj \Gamma {\Delta_1,\Delta_2,\Delta_3} {\letin{(x,y)}{s}{(\letin{\lunit}{t}{r})}=\letin{\lunit}{(\letin{(x,y)}{s}{t})}{r}} Z}
\and
\inferrule{\cj \Gamma {\Delta_1} t {Z\otimes W} \and \cj \Gamma {\Delta_2,z:Z,w:W} s {X\otimes Y} \and \cj \Gamma {\Delta_3,x:X,y:Y} r R }{\cj \Gamma {\Delta_1,\Delta_2,\Delta_3} {\letin{(x,y)}{(\letin{(z,w)}{t}{s}))}{r}=\letin{(z,w)}{t}{(\letin{(x,y)}{s}{r})}} R}
\end{mathparpagebreakable}
The additional equations for RMM terms in $\CCat$ are:
\begin{mathparpagebreakable}
    \inferrule{\cj \Gamma {\Delta_1} u {JA}
\and \cj {\Gamma} {\Delta_2,y:JA} t {TB}}{\cj \Gamma {\Delta_1,\Delta_2} {\doin{y}{\ret{u}}{t} = t[u/y]} {TB}}
    \and
    \inferrule{\cj \Gamma \Delta u {TA}
}{\cj \Gamma \Delta {\doin{y}{u}{\ret y} = u} {TA}}
    \and
    \inferrule{\cj \Gamma {\Delta_1} u {TA} \and \cj {\Gamma} {\Delta_2,y:JA} t {TB} \and \cj {\Gamma} {\Delta_3,z:JB} v {TC}}{\cj \Gamma {\Delta_1,\Delta_2,\Delta_3} {\doin{z}{(\doin{y}{u}{t})}{v}= \doin{y}{u}{(\doin{z}{t}{v})}} {TC}}
    \and
\inferrule{\cj \Gamma {\Delta_1} s {I} \and \cj \Gamma {\Delta_2} u {JA}}{\cj \Gamma {\Delta_1,\Delta_2} {\letin{\lunit}{s}{\ret{u}} = \ret{(\letin{\lunit}{s}{u})}} {TA}}
\and
\inferrule{\cj \Gamma {\Delta_1} s {X\otimes Y} \and \cj \Gamma {\Delta_2,x:X, y:Y} u {JA}}{\cj \Gamma {\Delta_1,\Delta_2} {\letin{(x,y)}{s}{\ret{u}} = \ret{(\letin{\lunit}{(x,y)}{u})}} {TA}}
\and
\inferrule{\cj \Gamma {\Delta_1} s {JA} \and \cj {\Gamma,a:A} {\Delta_2} u {JB}}{\cj \Gamma {\Delta_1,\Delta_2} {\letin{J(a)}{s}{\ret{u}} = \ret{(\letin{J(a)}{s}{u})}} {TB}}
    \and
\inferrule{\cj \Gamma {\Delta_1} s {I} \and \cj \Gamma {\Delta_2} u {TA}
\and \cj {\Gamma} {\Delta_3,y:JA} t {TB}}{\cj \Gamma {\Delta_1,\Delta_2, \Delta_3} {\doin{y}{(\letin{\lunit}{s}{u})}{t} = \letin{\lunit}{s}{(\doin{y}{u}{t})}} {TB}}
    \and
\inferrule{\cj \Gamma {\Delta_1} s {I} \and \cj \Gamma {\Delta_2} u {TA}
\and \cj {\Gamma} {\Delta_3,y:JA} t {TB}}{\cj \Gamma {\Delta_1,\Delta_2, \Delta_3} {\doin{y}{u}{(\letin{\lunit}{s}{t})} = \letin{\lunit}{s}{(\doin{y}{u}{t})}} {TB}}
    \and
\inferrule{\cj \Gamma {\Delta_1} s {X\otimes Y} \and \cj \Gamma {\Delta_2,x:X,y:Y} u {TA}
\and \cj {\Gamma} {\Delta_3,z:JA} t {TB}}{\cj \Gamma {\Delta_1,\Delta_2, \Delta_3} {\doin{z}{(\letin{(x,y)}{s}{u})}{t} = \letin{(x,y)}{s}{(\doin{z}{u}{t})}} {TB}}
    \and
\inferrule{\cj \Gamma {\Delta_1} s {X\otimes Y} \and \cj \Gamma {\Delta_2} u {TA}
\and \cj {\Gamma} {\Delta_3, x:X, y:Y, z:JA} t {TB}}{\cj \Gamma {\Delta_1,\Delta_2, \Delta_3} {\doin{z}{u}{(\letin{(x,y)}{s}{t})} = \letin{(x,y)}{s}{(\doin{z}{u}{t})}} {TB}}
    \and
\inferrule{\cj \Gamma {\Delta_1} s {JA} \and \cj {\Gamma,a:A} {\Delta_2} u {TB}
\and \cj {\Gamma} {\Delta_3,y:JB} t {TC}}{\cj \Gamma {\Delta_1,\Delta_2, \Delta_3} {\doin{y}{(\letin{J(a)}{s}{u})}{t} = \letin{J(a)}{s}{(\doin{y}{u}{t})}} {TC}}
    \and
\inferrule{\cj \Gamma {\Delta_1} s {JA} \and \cj {\Gamma} {\Delta_2} u {TB}
\and \cj {\Gamma,a:A} {\Delta_3,y:JB} t {TC}}{\cj \Gamma {\Delta_1,\Delta_2, \Delta_3} {\doin{y}{u}{(\letin{J(a)}{s}{u})} = \letin{J(a)}{s}{(\doin{y}{u}{t})}} {TC}}
    \and
    \inferrule{\cj \Gamma \Delta {t} {\underline{m}}}{\cj \Gamma \Delta {\underline{1_m} \; t = t} {\underline{m}}}
    \and
    \inferrule{\cj \Gamma \Delta {t} {\underline{m}}}{\cj \Gamma \Delta {\underline{\varphi \circ \xi} \; t = \underline{\varphi} \; ( \underline{\xi} \; t)} {\underline{l}}}(\xi \in \MCat(m,n),\varphi \in \MCat(n,l))
    \and
    \inferrule{\cj \Gamma \Delta {t} {\underline{e}}}{\cj \Gamma \Delta {\merg{\unmerg{t}} = t} {\underline{e}}}
     \and
    \inferrule{\cj \Gamma \Delta {t} {I}}{\cj \Gamma \Delta {\unmerg{\merg{t}} = t} {I}}
    \and
    \inferrule{\cj \Gamma \Delta {t} {\underline{m\mplus n}}}{\cj \Gamma \Delta {\merg{\unmerg{t}} = t} {\underline{m\mplus n}}}
    \and
    \inferrule{\cj \Gamma \Delta {t} {\underline{m}\otimes \underline{n}}}{\cj \Gamma \Delta {\unmerg{\merg{t}} = t} {\underline{m}\otimes \underline{n}}}
\end{mathparpagebreakable}
\begin{mathpar}
     \inferrule{\cj \Gamma \Delta {t} {\underline{m\mplus n}}}{\cj \Gamma \Delta {\unmerg{\app{(\underline{\xi\mplus \varphi}}{t})} = \letin{(s,r)}{\unmerg{t}}{(\app{\underline{\xi}}{s},\app{\underline{\varphi}}{r})}} {\underline{m'}\otimes \underline{n'}}} \text{\scriptsize $(\xi \in \MCat(m,m'),\varphi \in \MCat(n,n'))$}
\end{mathpar}
\begin{mathparpagebreakable}
    \inferrule{\cj \Gamma \Delta t {\underline{e\mplus m}}}{\cj \Gamma \Delta {\letin{(s,r)}{\unmerg{t}}{(\letin{\lunit}{\unmerg{s}}{r})}=\app{\underline{\lambda_m}}{t}} {\underline{m}}}
    \and
    \inferrule{\cj \Gamma \Delta t {\underline{m\mplus e}}}{\cj \Gamma \Delta {\letin{(s,r)}{\unmerg{t}}{(\letin{\lunit}{\unmerg{r}}{s})}=\app{\underline{\rho_m}}{t}} {\underline{m}}}
    \and
    \inferrule{\cj \Gamma \Delta t {\underline{m\mplus n}}}{\cj \Gamma \Delta {\letin{(s,r)}{\unmerg{t}}{(r,s)}=\unmerg{(\app{\underline{\sigma_{m,n}}}{t})}}{\underline{n}\otimes \underline{m}}}
    \and
    \inferrule{\cj \Gamma {\Delta_1} t {\underline{(m\mplus n)\mplus l}} \and {\cj \Gamma {\Delta_2, x:\underline{m},y:\underline{n},z:\underline{l}} u {\underline{m}\otimes (\underline{n}\otimes \underline{l})}}}{ \mbox{\(\begin{array}{l} \cj \Gamma {\Delta_1,\Delta_2} {\letin{(x,r)}{\unmerg{(\app{\underline{\alpha_{m,n,l}}}{t}})}{(\letin{(y,z)}{\unmerg{r}}{u})} =
    \\    \letin{(s,z)}{\unmerg{t}}{(\letin{(x,y)}{\unmerg{s}}{u})}}  {Z} \end{array}\)}}
     \and
\inferrule{\cj \Gamma {\Delta_1} s {I} \and \cj \Gamma {\Delta_2} u {\underline{e}}}{\cj \Gamma {\Delta_1,\Delta_2} {\letin{\lunit}{s}{\unmerg{u}} = \unmerg{(\letin{\lunit}{s}{u})}} {I}}
     \and
\inferrule{\cj \Gamma {\Delta_1} s {} \and \cj \Gamma {\Delta_2} u {\underline{m\oplus n}}}{\cj \Gamma {\Delta_1,\Delta_2} {\letin{\lunit}{s}{\unmerg{u}} = \unmerg{(\letin{\lunit}{s}{u})}} {\underline{m}\otimes \underline{n}}}
\and
\inferrule{\cj \Gamma {\Delta_1} s {X\otimes Y} \and \cj \Gamma {\Delta_2,x:X, y:Y} u {\underline{e}}}{\cj \Gamma {\Delta_1,\Delta_2} {\letin{(x,y)}{s}{\unmerg{u}} = \unmerg{(\letin{\lunit}{(x,y)}{u})}} {I}}
\and
\inferrule{\cj \Gamma {\Delta_1} s {X\otimes Y} \and \cj \Gamma {\Delta_2,x:X, y:Y} u {\underline{m \oplus n}}}{\cj \Gamma {\Delta_1,\Delta_2} {\letin{(x,y)}{s}{\unmerg{u}} = \unmerg{(\letin{\lunit}{(x,y)}{u})}} {\underline{m}\otimes \underline{n}}}
\and
\inferrule{\cj \Gamma {\Delta_1} s {JA} \and \cj {\Gamma,a:A} {\Delta_2} u {\underline{e}}}{\cj \Gamma {\Delta_1,\Delta_2} {\letin{J(a)}{s}{\unmerg{u}} = \unmerg{(\letin{J(a)}{s}{u})}} {I}}
\and
\inferrule{\cj \Gamma {\Delta_1} s {JA} \and \cj {\Gamma,a:A} {\Delta_2} u {\underline{m \oplus n}}}{\cj \Gamma {\Delta_1,\Delta_2} {\letin{J(a)}{s}{\unmerg{u}} = \unmerg{(\letin{J(a)}{s}{u})}} {\underline{m}\otimes \underline{n}}}
     \and
\inferrule{\cj \Gamma {\Delta_1} s {I} \and \cj \Gamma {\Delta_2} u {I}}{\cj \Gamma {\Delta_1,\Delta_2} {\letin{\lunit}{s}{\merg{u}} = \merg{(\letin{\lunit}{s}{u})}} {\underline{e}}}
     \and
\inferrule{\cj \Gamma {\Delta_1} s {} \and \cj \Gamma {\Delta_2} u {\underline{m}\otimes \underline{n}}}{\cj \Gamma {\Delta_1,\Delta_2} {\letin{\lunit}{s}{\unmerg{u}} = \unmerg{(\letin{\lunit}{s}{u})}} {\underline{m\oplus n}}}
\and
\inferrule{\cj \Gamma {\Delta_1} s {X\otimes Y} \and \cj \Gamma {\Delta_2,x:X, y:Y} u {I}}{\cj \Gamma {\Delta_1,\Delta_2} {\letin{(x,y)}{s}{\unmerg{u}} = \unmerg{(\letin{\lunit}{(x,y)}{u})}} {\underline{e}}}
\and
\inferrule{\cj \Gamma {\Delta_1} s {X\otimes Y} \and \cj \Gamma {\Delta_2,x:X, y:Y} u {\underline{m}\otimes \underline{n}}}{\cj \Gamma {\Delta_1,\Delta_2} {\letin{(x,y)}{s}{\unmerg{u}} = \unmerg{(\letin{\lunit}{(x,y)}{u})}} {\underline{m\oplus n}}}
\and
\inferrule{\cj \Gamma {\Delta_1} s {JA} \and \cj {\Gamma,a:A} {\Delta_2} u {I}}{\cj \Gamma {\Delta_1,\Delta_2} {\letin{J(a)}{s}{\unmerg{u}} = \unmerg{(\letin{J(a)}{s}{u})}} {\underline{e}}}
\and
\inferrule{\cj \Gamma {\Delta_1} s {JA} \and \cj {\Gamma,a:A} {\Delta_2} u {\underline{m}\otimes \underline{n}}}{\cj \Gamma {\Delta_1,\Delta_2} {\letin{J(a)}{s}{\unmerg{u}} = \unmerg{(\letin{J(a)}{s}{u})}} {\underline{m\oplus n}}}
\end{mathparpagebreakable}
The equations for linear-non-linear interactions are:
\begin{mathparpagebreakable}
 \inferrule{\aj \Gamma {u} {A} \and \cj {\Gamma,a:A} {\Delta} {t} X}{ \cj \Gamma {\Delta} {\letin{J(a)}{J(u)}{t}=t[u/a]} X}
 \and
  \inferrule{\cj \Gamma \Delta {t} {JA}}{ \cj \Gamma {\Delta} {\letin{J(a)}{t}{J(a)}=t} JA}
\and
 \inferrule{\cj {\Gamma} {\emptyset} {t} X}{ \cj \Gamma \emptyset {\derek R(t)=t} X}
 \and
 \inferrule{\aj {\Gamma}  {u} RX}{\aj \Gamma {R(\derek{u})=u} RX}
 \and
 \inferrule{\cj \Gamma {\Delta_1} t {I} \and \cj \Gamma {\Delta_2} s {JA} \and \cj {\Gamma,a:A} {\Delta_3} r Z}{\cj \Gamma {\Delta_1,\Delta_2,\Delta_3} {\letin{J(a)}{(\letin{\lunit}{t}{s})}{r}=\letin{\lunit}{t}{(\letin{J(a)}{s}{r})}} Z}
  \and
 \inferrule{\cj \Gamma {\Delta_1} t {I} \and \cj \Gamma {\Delta_2} s {JA} \and \cj {\Gamma,a:A} {\Delta_3} r Z}{\cj \Gamma {\Delta_1,\Delta_2,\Delta_3} {\letin{J(a)}{s}{(\letin{\lunit}{t}{r})}=\letin{\lunit}{t}{(\letin{J(a)}{s}{r})}} Z}
 \and
  \inferrule{\cj {\Gamma,a:A} {\Delta_1} t {I} \and \cj \Gamma {\Delta_2} s {JA} \and \cj \Gamma {\Delta_3} r Z}{\cj \Gamma {\Delta_1,\Delta_2,\Delta_3} {\letin{J(a)}{s}{(\letin{\lunit}{t}{r})}=\letin{\lunit}{(\letin{J(a)}{s}{t})}{r}} Z}
   \and
 \inferrule{\cj \Gamma {\Delta_1} t {X\otimes Y} \and \cj \Gamma {\Delta_2,x:X,y:Y} s {JA} \and \cj {\Gamma,a:A} {\Delta_3} r Z}{\cj \Gamma {\Delta_1,\Delta_2,\Delta_3} {\letin{J(a)}{(\letin{(x,y)}{t}{s})}{r}=\letin{(x,y)}{t}{(\letin{J(a)}{s}{r})}} Z}
    \and
 \inferrule{\cj \Gamma {\Delta_1} t {X\otimes Y} \and \cj \Gamma {\Delta_2} s {JA} \and \cj {\Gamma,a:A} {\Delta_3,x:X,y:Y} r Z}{\cj \Gamma {\Delta_1,\Delta_2,\Delta_3} {\letin{J(a)}{s}{(\letin{(x,y)}{t}{r})}=\letin{(x,y)}{t}{(\letin{J(a)}{s}{r})}} Z}
\and
  \inferrule{\cj {\Gamma,a:A} {\Delta_1} t {X\otimes Y} \and \cj \Gamma {\Delta_2} s {JA} \and \cj \Gamma {\Delta_3,x:X,y:Y} r Z}{\cj \Gamma {\Delta_1,\Delta_2,\Delta_3} {\letin{J(a)}{s}{(\letin{(x,y)}{t}{r})}=\letin{(x,y)}{(\letin{J(a)}{s}{t})}{r}} Z}
\and
 \inferrule{\cj {\Gamma,a:A} {\Delta_1} t {X\multimap Y} \and \cj \Gamma {\Delta_2} s {JA} \and \cj \Gamma {\Delta_3} r X}{\cj \Gamma {\Delta_1,\Delta_2,\Delta_3} {\app{(\letin{J(a)}{s}{t})}{r}=\letin{J(a)}{s}{(\app{t}{r})} }Y}
 \and
 \inferrule{\cj {\Gamma} {\Delta_1} t {X\multimap Y} \and \cj \Gamma {\Delta_2} s {JA} \and \cj \Gamma {\Delta_3,a:A} r X}{\cj \Gamma {\Delta_1,\Delta_2,\Delta_3} {\app{t}{(\letin{J(a)}{s}{r})}=\letin{J(a)}{s}{(\app{t}{r})} }Y}
 \and
 \inferrule{\cj {\Gamma,a:A} {\Delta_1,x:X} t {Y} \and \cj \Gamma {\Delta_2} s {JA}}{\cj \Gamma {\Delta_1,\Delta_2} {\letin{J(a)}{s}{\lambda (x:X). t}={\lambda(x:X).(\letin{J(a)}{s}{t})}} {X\multimap Y}}
\and
 \inferrule{\cj \Gamma {\Delta_1} t {JA} \and \cj {\Gamma,a:A} {\Delta_2} s {JB} \and \cj {\Gamma,b:B} {\Delta_3} r X}{\cj \Gamma {\Delta_1,\Delta_2,\Delta_3} {\letin{J(b)}{(\letin{J(a)}{t}{s})}{r}=\letin{J(a)}{t}{(\letin{J(b)}{s}{r})} }X}
\end{mathparpagebreakable}

\subsection{Arrow relative monadic metalanguage (ARMM)} \label{app:armm-equations}

The types and contexts of ARMM are:
\begin{alignat*}{4}
    &\text{$\ACat$ types:} \quad & A,B &::= 1 ~|~ A\times B ~|~ A\Rightarrow X & \\
    &\text{$\CCat$ types:} \quad & X,Y &::= 1 ~|~ X\times Y ~|~ JA ~|~ KA ~|~ TA & \\
    &\text{$\ACat$ contexts:} \quad & \Gamma, \Delta &::= \emptyset ~|~ \Gamma, a:A & \\
    &\text{$\CCat$ contexts:} \quad & \Phi &::= \emptyset ~|~ \Phi, x:X &
\end{alignat*}
The terms for $\ACat$ types are:
\begin{mathparpagebreakable}
    \inferrule{ }{\aj {\Gamma,a:A} a {A}}
    \and
    \inferrule{ }{\aj {\Gamma} \lunit {1}}
    \and
    \inferrule{\aj \Gamma u {A\times B}}{\tj \Gamma {\pi_1 u} {A}}
    \and
    \inferrule{\aj \Gamma u {A\times B}}{\tj \Gamma {\pi_2 u} {B}}
    \and
    \inferrule{\aj {\Gamma} u {A} \and \tj {\Gamma} v {B}}{\tj {\Gamma} {(u,v)} {A\times B}}
\end{mathparpagebreakable}
The terms for $\CCat$ types are:
\begin{mathparpagebreakable}
    \inferrule{ }{\ej {\Gamma} \Delta {\Phi, x:X} x {X}}
    \and
    \inferrule{ }{\ej {\Gamma} \Delta \Phi \lunit {1}}
    \and
    \inferrule{\ej \Gamma \Delta \Phi t {X\times Y}}{\ej \Gamma \Delta \Phi {\pi_1 t} {X}}
    \and
    \inferrule{\ej \Gamma \Delta \Phi t {X\times Y}}{\ej \Gamma \Delta \Phi {\pi_2 t} {Y}}
    \and
    \inferrule{\ej {\Gamma} \Delta \Phi t {X} \and \ej {\Gamma} \Delta \Phi s {Y}}{\ej {\Gamma} \Delta \Phi {(t,s)} {X\times Y}}
    \and
    \inferrule{\ej \Gamma \Delta \Phi {t} {JA} 
    \and
    \ej {\Gamma} {\Delta, a:A} {\Phi} {s} {X}}{
   \ej \Gamma \Delta {\Phi} {\letin{J(a)}{t}{s}} X
    }
    \and
    \inferrule{\ej \Gamma \Delta \Phi {t} {KA} 
    \and
    \ej {\Gamma, a:A} {\Delta} {\Phi} {t} {X}}{
   \ej \Gamma \Delta {\Phi} {\letin{K(a)}{t}{s}} X
    }
    \and
    \inferrule{\ej \Gamma \Delta \Phi t {JA}}{\ej \Gamma \Delta \Phi {\ret{t}} {TA}}
    \and
    \inferrule{\ej \Gamma \Delta {\Phi} t {TA} 
    \and
    \ej {\Gamma} {\Delta} {x:JA} s {TB}
    }{
    \ej \Gamma \Delta \Phi {\doin{x}{t}{s}} {TB}
    }
\end{mathparpagebreakable}
The terms for the interaction between $\ACat$ and $\CCat$ are:
\begin{mathparpagebreakable}
    \inferrule{\aj \Gamma u {A\Rightarrow X} \and \aj {\Gamma,\Delta} {v} A}{\ej \Gamma \Delta \Phi {\aapp u v} {X}}
    \and
    \inferrule{\ej \Gamma {x:A} \emptyset t X}{\aj \Gamma {\aabs x {A} t} {A\Rightarrow X}}
    \\
    \inferrule{\aj {\Gamma} u A}{\ej \Gamma \Delta \Phi {K(u)} {KA}}
    \and
    \inferrule{\aj {\Gamma,\Delta} u A}{\ej {\Gamma} {\Delta} \Phi {J(u)} {JA}}
    \and
    \inferrule{\aj {\Gamma} u {A \Rightarrow JB} \and
    \aj \Gamma v A}{\aj {\Gamma} {\fapp{u}{v}} {B}}
\end{mathparpagebreakable}
The equations for terms in $\ACat$ are:
\begin{mathparpagebreakable}
    \inferrule{\aj \Gamma u 1 }{\aj \Gamma {u= \lunit} 1}
    \and
    \inferrule{\aj \Gamma {u} {A} \and \aj \Gamma {v} {B}}{\aj \Gamma {\pi_1 (u,v)=u} {A}}
    \and
    \inferrule{\aj \Gamma {u} {A} \and \aj \Gamma {v} {B}}{\aj \Gamma {\pi_2 (u,v)=t} {B}}
    \and
    \inferrule{\aj \Gamma {u} {A\times B}}{\aj \Gamma {(\pi_1 u,\pi_2 u)=u} {A\times B}}
\end{mathparpagebreakable}
The equations for terms in $\CCat$ are:
\begin{mathparpagebreakable}
    \inferrule{\ej \Gamma \Delta \Phi t 1 }{\ej \Gamma \Delta \Phi {t = \lunit} 1}
    \and
    \inferrule{\ej \Gamma \Delta \Phi {t} {X} \and \ej \Gamma \Delta \Phi {s} {Y}}{\ej \Gamma \Delta \Phi {\pi_1 (t,s)=t} {X}}
    \and
    \inferrule{\ej \Gamma \Delta \Phi {t} {X} \and \ej \Gamma \Delta \Phi {s} {Y}}{\ej \Gamma \Delta \Phi {\pi_2 (t,s)=s} {Y}}
    \and
    \inferrule{\ej \Gamma \Delta \Phi{t} {X\times Y}}{\ej \Gamma \Delta \Phi {(\pi_1 t,\pi_2 t)=t} {X\times Y}}
    \and
    \inferrule{\ej \Gamma \Delta \Phi t {JA}
    \and
    \ej \Gamma \Delta {x:JA} s TB}
    {\ej \Gamma \Delta \Phi {\doin{x}{\ret{t}}{s} \ = \ s[t / x]} {TB}}
    \and
    \inferrule{\ej \Gamma \Delta {\Phi} t {TA}}
    {\ej \Gamma \Delta \Phi {\doin{x}{t}{\ret{x}} \ = \ t} {TB}}
    \and
    \inferrule{\ej \Gamma \Delta \Phi r {TA}
    \and
    \ej \Gamma \Delta {x:JA} s {TB}
    \and
    \ej \Gamma \Delta {y:JB} t {TC}}
    {\ej \Gamma \Delta \Phi {\doin{x}{r}{(\doin{y}{s}{t})} \ = \ \doin{y}{(\doin{x}{r}{s})}{t}} {TC}}
    \and
    \inferrule{\ej \Gamma \Delta \Phi r {JA}
    \and
    \ej \Gamma {\Delta,a:A} \Phi s {JB}
    \and
    \ej \Gamma {\Delta, b:B} \Phi t {X}}
    {\ej \Gamma \Delta \Phi {\letin{J(a)}{r}{(\letin{J(b)}{s}{t})} = \letin{J(b)}{(\letin{J(a)}{r}{s})}{t}} {X}}
    \and
    \inferrule{\ej \Gamma \Delta \Phi r {KA}
    \and
    \ej {\Gamma,a:A} {\Delta} \Phi s {KB}
    \and
    \ej {\Gamma,b:B} {\Delta} \Phi t {X}}
    {\ej \Gamma \Delta \Phi {\letin{K(a)}{r}{(\letin{K(b)}{s}{t})} = \letin{K(b)}{(\letin{K(a)}{r}{s})}{t}} {X}}
    \and
    \inferrule{\ej \Gamma \Delta \Phi r {JA}
    \and
    \ej {\Gamma} {\Delta,a:A} \Phi s {KB}
    \and
    \ej {\Gamma,b:B} {\Delta} \Phi t {X}}
    {\ej \Gamma \Delta \Phi {\letin{J(a)}{r}{(\letin{K(b)}{s}{t})} = \letin{K(b)}{(\letin{J(a)}{r}{s})}{t}} {X}}
    \and
    \inferrule{\ej \Gamma \Delta \Phi r {KA}
    \and
    \ej {\Gamma,a:A} {\Delta} \Phi s {JB}
    \and
    \ej {\Gamma} {\Delta,b:B} \Phi t {X}}
    {\ej \Gamma \Delta \Phi {\letin{K(a)}{r}{(\letin{J(b)}{s}{t})} = \letin{J(b)}{(\letin{K(a)}{r}{s})}{t}} {X}}
    \and
    \inferrule{\ej \Gamma \Delta {\Phi} {u} {J1}}{\ej \Gamma \Delta {\Phi} {u = J\lunit} {J1}}
    \and
    \inferrule{\ej \Gamma \Delta {\Phi} {u} {K1} }{\ej \Gamma \Delta {\Phi} {u = K\lunit} {K1}}
    \and
    \inferrule{\ej \Gamma \Delta \Phi r {JA}
    \and
    \ej {\Gamma} {\Delta, a:A} \Phi s {X}
    \and
    \ej {\Gamma} {\Delta, a:A} \Phi t {Y}}
    {\ej \Gamma \Delta \Phi {\letin{J(a)}{r}{(s,t)} = (\letin{J(a)}{r}{s},\letin{J(a)}{r}{t})} {X\times Y}}
    \and
    \inferrule{\ej \Gamma \Delta \Phi s {JA} \and \ej \Gamma \Delta \Phi t {X}}
    {\ej \Gamma \Delta \Phi {\letin{J(a)}{s}{t} = t} {X}}(a : A \notin \mathsf{fv}(t))
    \and
    \inferrule{\ej \Gamma \Delta \Phi r {KA}
    \and
    \ej {\Gamma, a:A} {\Delta} \Phi s {X}
    \and
    \ej {\Gamma, a:A} {\Delta} \Phi t {Y}}
    {\ej \Gamma \Delta \Phi {\letin{K(a)}{r}{(s,t)} = (\letin{K(a)}{r}{s},\letin{K(a)}{r}{t})} {X\times Y}}
    \and
    \inferrule{\ej \Gamma \Delta \Phi s {KA} \and \ej \Gamma \Delta \Phi t {X}}
    {\ej \Gamma \Delta \Phi {\letin{K(a)}{s}{t} = t} {X}}(a:A \notin \mathsf{fv}(t))
    \and
    \inferrule{\ej \Gamma \Delta \Phi r {JA}
    \and
    \ej {\Gamma} {\Delta} \Phi s {JB}
    \and
    \ej {\Gamma} {\Delta, a:A,b:B} \Phi t {X}}
    {\ej \Gamma \Delta \Phi {\letin{J(a)}{r}{(\letin{J(b)}{s}{t})} = \letin{J(b)}{s}{(\letin{J(a)}{r}{t})}} {X}}
    \and
    \inferrule{\ej \Gamma \Delta \Phi r {KA}
    \and
    \ej {\Gamma} {\Delta} \Phi s {KB}
    \and
    \ej {\Gamma, a:A,b:B} {\Delta} \Phi t {X}}
    {\ej \Gamma \Delta \Phi {\letin{K(a)}{r}{(\letin{K(b)}{s}{t})} = \letin{K(b)}{s}{(\letin{K(a)}{r}{t})}} {X}}
    \and
    \inferrule{\ej \Gamma \Delta \Phi s {JA}
    \and
    \ej {\Gamma} {\Delta, a:A,a':A} \Phi t {X}}
    {\ej \Gamma \Delta \Phi {\letin{J(a)}{s}{(\letin{J(a')}{s}{t})} = \letin{J(a)}{s}{(t[a/a'])}} {X}}
    \and
    \inferrule{\ej \Gamma \Delta \Phi s {KA}
    \and
    \ej {\Gamma, a:A,a':A} {\Delta} \Phi t {X}}
    {\ej \Gamma \Delta \Phi {\letin{K(a)}{s}{(\letin{K(a')}{s}{t})} = \letin{K(a)}{s}{(t[a/a'])}} {X}}
\end{mathparpagebreakable}
The equations for the interaction between $\ACat$ and $\CCat$ are:
\begin{mathparpagebreakable}
    \inferrule{\aj \Delta {u} {A} \and \ej \Gamma {\Delta,a:A} {\Phi} {t} X}{\ej \Gamma {\Delta} \Phi {\letin{J(a)}{J(u)}{t}=t[u/a]} X}
    \and
    \inferrule{\ej \Gamma \Delta \Phi {t} {JA}}{\ej \Gamma {\Delta} \Phi {\letin{J(a)}{t}{J(a)}=t} JA}
    \and
    \inferrule{\aj \Gamma {u} {A} \and \ej {\Gamma,a:A} {\Delta} {\Phi} {t} X}{\ej \Gamma {\Delta} \Phi {\letin{K(a)}{K(u)}{t}=t[u/a]} X}
    \and
    \inferrule{\ej \Gamma \Delta \Phi {t} {KA}}{\ej \Gamma {\Delta} \Phi {\letin{K(a)}{t}{K(a)}=t} KA}
    \and
    \inferrule{\ej {\Gamma} {a:A} \emptyset t X \and \aj {\Gamma} {u} {A}}{\ej {\Gamma} {\emptyset} \emptyset {\aapp{(\aabs{a}{A}{t})}{u} = t[u / a]} X}
    \and
    \inferrule{\aj \Gamma {u} {A\Rightarrow X}}{\aj \Gamma {\aabs{a}{A}{\aapp{u}{a}} = u} {A\Rightarrow X}}
    \and
    \inferrule{\aj {\Gamma,a:A} u {B} \and \aj \Gamma v A}{\aj  \Gamma {\fapp{(\aabs{a}{A}{J(u)})}{v} = u[v / a]} B}
    \and
    \inferrule{\aj {\Gamma} u {A \Rightarrow JB}}{\aj {\Gamma} {\aabs{a}{A}{J(\fapp{u}{a})} = u} {A \Rightarrow JB}}
\end{mathparpagebreakable}
\end{document}